\documentclass[12pt]{article}

\newcommand{\QR}{\mathrm{QR}}

\usepackage[most]{tcolorbox}
\newtcolorbox{yellowbox}{colback=yellow!20, colframe=yellow!80!black}

\newtcolorbox{redbox}{
  colback=red!5!white,
  colframe=red!50!black
}

\usepackage{soul}
\usepackage{changepage}

\usepackage[margin=1.125in]{geometry}
\usepackage{xcolor}
\definecolor{DarkGreen}{rgb}{0.1,0.5,0.1}
\definecolor{DarkRed}{rgb}{0.5,0.1,0.1}
\definecolor{DarkBlue}{rgb}{0.1,0.1,0.5}

\usepackage[small]{caption}
\usepackage[pdftex]{hyperref}
\hypersetup{
    unicode=false,          % non-Latin characters in Acrobat¿s bookmarks
    pdftoolbar=true,        % show Acrobat toolbar?
    pdfmenubar=true,        % show Acrobat menu?
    pdffitwindow=false,      % page fit to window when opened
    pdfnewwindow=true,      % links in new window
    colorlinks=true,       % false: boxed links; true: colored links
    linkcolor=DarkBlue,          % color of internal links
    citecolor=DarkGreen,        % color of links to bibliography
    filecolor=DarkGreen,      % color of file links
    urlcolor=DarkBlue,          % color of external links
    pdftitle={},
    pdfauthor={},    
}
\usepackage{amssymb,amsmath,amsthm,amsfonts,enumitem,mathtools}
\usepackage{blkarray}
\usepackage{fullpage,nicefrac,comment}
\usepackage{tikz}
\usetikzlibrary{arrows,decorations.pathmorphing,decorations.shapes,patterns,positioning}
\usepackage[ruled,linesnumbered]{algorithm2e}

\newcommand{\argmax}{\mathrm{argmax}}

\newcommand{\defeq}{\vcentcolon=}

\newcommand{\eps}{\varepsilon}
\renewcommand{\epsilon}{\varepsilon}

\newtheorem{theorem}{Theorem}  
\newtheorem*{theorem*}{Theorem}
\newtheorem{lemma}[theorem]{Lemma} 
\newtheorem{assumption}[theorem]{Assumption} 
\newtheorem*{lemma*}{Lemma}
\newtheorem{definition}[theorem]{Definition}

\newtheorem{observation}[theorem]{Observation}

\newtheorem{corollary}[theorem]{Corollary} 

\newtheorem{remark}{Remark}
\newtheorem{claim}[theorem]{Claim}

\usepackage{thmtools}
\usepackage{thm-restate}
\usepackage[capitalise]{cleveref}

\title{Limitations to Computing Quadratic Functions on Reed-Solomon Encoded Data}
\author{Keller Blackwell\thanks{Stanford University. Email: kellerb@stanford.edu.} \ and Mary Wootters\thanks{Stanford University. Email: marykw@stanford.edu.}}

\begin{document}
\maketitle

\begin{abstract}
    We study the problem of low-bandwidth non-linear computation on Reed-Solomon encoded data. Given an $[n,k]$ Reed-Solomon encoding of a message vector $\mathbf{f} \in \mathbb{F}_q^k$, and a polynomial $g \in \mathbb{F}_q[X_1, X_2, \ldots, X_k]$, a user wishing to evaluate $g(\mathbf{f})$ is given local query access to each codeword symbol. The query response is allowed to be the output of an arbitrary function evaluated locally on the codeword symbol, and the user's aim is to minimize the total information downloaded in order to compute $g(\mathbf{f})$.  This problem has been studied before for \emph{linear} functions $g$; in this work we initiate the study of non-linear functions by starting with quadratic monomials.

    For $q = p^e$ and distinct $i,j \in [k]$, we show that any scheme evaluating the quadratic monomial $g_{i,j} \defeq X_i X_j$ must download at least $2 \log_2(q-1) - 3$ bits of information when $p$ is an odd prime, and at least $2\log_2(q-2) -4$ bits when $p=2$. When $k=2$, our result shows that one cannot do significantly better than the naive bound of $k \log_2(q)$ bits, which is enough to recover all of $\mathbf{f}$. This contrasts sharply with prior work for low-bandwidth evaluation of \emph{linear} functions $g(\mathbf{f})$ over Reed-Solomon encoded data, for which it is possible to substantially improve upon this bound~\cite{GW16,TYB18,SW21,Kiah24,CT22}.

\end{abstract}

\section{Introduction}
Suppose that data $\mathbf{f} \in \mathbb{F}_q^k$ is encoded with an error correcting code to produce a vector $\mathbf{c} \in \mathbb{F}_q^n$, for some $n > k$.  The problem of \emph{low-bandwidth computation} on top of the error correction is to compute some function $g(\mathbf{f})$, given access to $\mathbf{c}$, with limited \emph{bandwidth}.  That is, for some parameter $s \in \mathbb{Z}^+$, given an arbitrary $I \subseteq [n]$ of size $|I| = s$, we are allowed to query an arbitrary function $\lambda_i(c_i)$ of the symbol of $c_i$ for each $i \in I$.  The goal is to compute $g(\mathbf{f})$, while minimizing the number of bits queried (that is, the sum of $ \log_2(|\lambda_i(c_i)|)$ over all $i \in I$).

Variants of this problem arise organically in many domains,  including \emph{regenerating codes} in distributed storage~(e.g., \cite{DGW10,dimakis_survey}); \emph{homomorphic secret sharing} (e.g, \cite{C:Benaloh86a, BGI16a, BGILT18, FIKW22}) and 
\emph{low-bandwidth secret sharing} (e.g., \cite{Wang2008OnSR, Huang16, HB17, ZHANG2012106}) in secret sharing;  and \emph{coded computation} in distributed computing (e.g., \cite{LLPPR18,DCG19,YLRKSA19}); these connections---and implications of our work in these domains---are elaborated in Section~\ref{sec: related works}.

We study the case when the error correcting code is a Reed-Solomon (RS) code. In RS codes, the codeword symbols are indexed by $n$ evaluation points $\alpha \in \mathbb{F}_q$; the data $\mathbf{f} \in \mathbb{F}_q^k$ is interpreted as a polynomial $f \in \mathbb{F}_q[x]$ of degree at most $k-1$, given by $f(x) = \sum_{i=0}^{k-1} {f}_i x^i$. The corresponding encoding $\mathbf{c} \in \mathbb{F}_q^n$, called a \emph{codeword}, is indexed by $n$ distinct evaluation points $\alpha \in \mathbb{F}_q$, and is given by ${c}_\alpha = f(\alpha)$.  Reed-Solomon codes are a classical error correcting code, ubiquitous in both theory and practice. Relevant to the domains mentioned above, RS codes are used in distributed storage (e.g.,~\cite{HDFS,Ceph}); in secret sharing as \emph{Shamir's scheme}~\cite{Shamir79}; and for coded computation (e.g., \cite{YMA17,YLRKSA19}). When the error correcting code is any MDS\footnote{MDS, or \emph{Maximum Distance Separable} codes, are codes of dimension $k$ and length $n$ with the best possible distance $n-k+1$.  In particular, in an MDS code, any $k$ symbols of the codeword uniquely determine the message.} code (including an RS code), the naive approach to compute $g(\mathbf{f})$ is to first recover $\mathbf{f}$ entirely by querying any $k$ symbols of the codeword in full; one can then compute $g(\mathbf{f})$ for any function $g$. This approach has bandwidth $k  \log_2(q) $ bits. 

The natural question is whether one can do better.  For RS codes, prior work has shown that the answer is \emph{yes} when $g:\mathbb{F}_q^k \to \mathbb{F}_q$ is a linear function. \emph{Regenerating codes} capture the special case of this problem where $g(\mathbf{f})$ is the function $g_\alpha(\mathbf{f}) = f(\alpha)$ for some $\alpha \in \mathbb{F}_q$.  A long line of work~(e.g., \cite{Shanmugam14,GW17,TYB18,CT22}) has established that one can compute $g_\alpha(\mathbf{f})$ for any $\alpha$ that appears as an evaluation point in the RS code, using substantially fewer than $k \log_2(q)$ bits in bandwidth.  In certain parameter regimes~\cite{TYB18,CT22}, the bandwidth can even get close to $\log_2(q)$ bits, the minimum number of bits required to represent $g_\alpha(\mathbf{f}) \in \mathbb{F}_q$.

For arbitrary linear functions $g:\mathbb{F}_q^k \to \mathbb{F}_q$, it is again possible to use asymptotically less than $k \log_2(q)$ bits, at least in some parameter regimes~\cite{SW21, Kiah24}.  \cite{SW21} showed that, given a full-length RS code of dimension $k = (1-\eps)n$ over an extension field and query access to $(1 - \gamma)n$ nodes for any $\gamma < \eps$, one can evaluation \emph{any} linear $g(\mathbf{f})$ with bandwidth $O(n/(\eps - \gamma))$.  When $\eps, \gamma$ are constants, this is a factor of $\Omega(\log n)$ less than the naive bound. The work \cite{Kiah24} extends the techniques of \cite{SW21} to consider linear $g(\mathbf{f})$ which are some linear combination of $\ell \leq k$ codeword symbols; that is, when $g(\mathbf{f}) = \sum_{j=1}^\ell \kappa_j f(\alpha_{i_j})$ for some choice of $i_1, \ldots, i_\ell \in [n]$. For RS codes over $\mathbb{F}_q=\mathbb{F}_{p^e}$, \cite{Kiah24} downloads $O(\ell p^{e-1} \log_2(p))$ bits, outperforming the naive $k \log_2(q) = k e \log_2(p)$ when $\ell \ll k$ and $p,e$ are small.

These results demonstrate that it is possible to improve on the naive scheme when evaluating  \emph{linear} functions $g(\mathbf{f})$. The next natural question is whether this is possible for \emph{non-linear functions}. That is, our question is:

\begin{center}
\begin{adjustwidth}{2cm}{2cm}
\itshape 
\centering
What is the minimal necessary bandwidth to compute \textbf{non-linear} functions $g : \mathbb{F}_q^k \to \mathbb{F}_q$ of Reed-Solomon encoded data?
\end{adjustwidth}
\end{center}

\paragraph{Main Result in a Nutshell.} We study the simplest instance of non-linear functions: computing quadratic monomials $g$ on top of dimension $k$ Reed-Solomon codes.  We work over arbitrary finite fields $\mathbb{F}_q = \mathbb{F}_{p^e}$. As discussed above, the data $\mathbf{f} = (f_0, f_1, \ldots, f_{k-1})$ represents a degree $\deg(f) \leq k-1$ polynomial over $\mathbb{F}_q$, and we consider the task of computing a quadratic monomial $g_{i,j}(\mathbf{f})\defeq f_i f_j$, for  $i,j \in [0,k-1]$. 

In this setting, one might hope to be able to do better than $2 \log_2 (q)$, the number of bits needed to represent both $f_i$ and $f_j$ separately; the goal would be to get closer to $\log_2(q)$, the number of bits needed to represent $f_i \cdot f_j$.  However, we show that this is not possible!

\vspace{.5cm}
Our main result,  Theorem~\ref{thm: full-case bound} below, implies that for all $i,j \in [0,k-1]$, $i \neq j$, any scheme computing $g_{i,j}(\mathbf{f}) = f_if_j$ must download at least $2  \log_2(q-1)  - 3$ bits when $p$ is an odd prime; or $2 \log_2(q-2) - 4$ bits when $p=2$.\footnote{Note that such a scheme need not compute \emph{every} quadratic monomial; just being able to compute one such monomial suffices for the lower bound to hold.} 
    This is nearly the full $2 \log_2(q)$ bits needed to represent {both} $f_i$ and $f_j$ separately.
    While the lower bound holds for all $k$, when $k=2$ this bound implies the impossibility of computing quadratic functions with download bandwidth even a few bits less than the naive bound of $k \log_2(q)$.

\vspace{0.5cm}

We view our results---which we state in more detail in Section~\ref{sec:results} below---as an important first step towards addressing the question above about general nonlinear computation.  When $k=2$, our results have the surprising implication that one cannot meaningfully improve on the naive bound of $k\log_2(q)$ for quadratic monomials, in contrast with the case for linear functions.  As discussed more in Sections~\ref{sec:results} and \ref{sec: related works}, our results also shed interesting light in related domains, including regenerating codes, low-bandwidth secret sharing, leakage resilience, and homomorphic secret sharing. 

\subsection{Quadratic Monomial (QM) Recovery}

We now formalize our problem. We consider the simplest setting for non-linear evaluation, which is computing quadratic monomials on top of Reed-Solomon codes. 

\begin{definition}[Reed-Solomon (RS) codes of dimension $k$ \cite{RS60}]
    Let $\mathbb{F}_q$ denote the finite field of order $q$ and let $n=q$. The corresponding (full length) Reed-Solomon code of dimension $k$ is the vector space
    \begin{equation*}
        \mathrm{RS}_q[n,k] \defeq \left\lbrace \langle f(\alpha) \rangle_{\alpha \in \mathbb{F}_q} : f \in \mathbb{F}_q[x], \deg(f) < k \right\rbrace. 
    \end{equation*}
\end{definition}

A dimension-$k$ RS code encodes a message vector $\mathbf{f} = (f_0, f_1, \ldots, f_{k-1})$ as evaluations of the polynomial $f(x) = \sum_{i=0}^{k-1} f_i x^i$. For some $i, j \in [0,k-1]$, we wish to compute $g_{i,j}(\mathbf{f}) \defeq f_i f_j$. To formalize the model described informally above, we first define a \emph{leakage function}, which outputs a single bit.  

\begin{definition}[Leakage Function]\label{def: leakage function}
    For $A \subseteq \mathbb{F}_q$, the (bit-valued) leakage function $\lambda = \lambda(A): \mathbb{F}_q \to \lbrace 0, 1 \rbrace$ is given by 
    \begin{equation*}
        \lambda(x) = \begin{cases}
            0 & x \in A\\
            1 & \text{else}
        \end{cases}
    \end{equation*}
\end{definition}

Each server may evaluate any (non-negative) number of leakage functions as part of a scheme to compute $g_{i,j}(\mathbf{f})$.  We formalize this as follows.

\begin{definition}[Quadratic monomial recovery]\label{def: QM}
    Let $t,s,k\in \mathbb{Z}^+$ and $i,j \in [0,k-1]$. We say that there exists a $t$-bit, $s$-server \textit{\textbf{Q}uadratic \textbf{M}onomial recovery scheme} (\textbf{QM}) for $g_{i,j}$ and RS codes of dimension $k$ if for every choice of $S  \subseteq \mathbb{F}_q$ with $|S| = s$, there exists
    
    \begin{itemize}
        \item a sequence of $\alpha_1, \ldots, \alpha_t \in S$, not necessarily distinct;
        \item leakage functions $\lambda_z: \mathbb{F}_q \to \lbrace 0, 1 \rbrace$, $z \in [t]$;
        \item and a reconstructing function $\mathrm{Rec}: \lbrace 0, 1 \rbrace^t \to \mathbb{F}_q$
    \end{itemize}
    such that 
    \begin{equation}
        \mathrm{Rec}\left( \lambda_{z} \left( f(\alpha_z) \right)  :  z \in [t] \right) = g_{i,j}(\mathbf{f})
    \end{equation}
    for all $f \in \mathbb{F}_q[x]$, $\deg(f) \leq k-1$. Given a $t$-bit, $s$-server QM scheme, we call the parameter $t$ the download bandwidth of the scheme.
\end{definition}

\subsection{Our Results}\label{sec:results}

Our main result is the following.

\begin{restatable}{theorem}{fullcasebound}\emph{(Main Theorem) }
\label{thm: full-case bound}Let $k \geq 2$ and fix $i,j \in [0,k-1]$, $i\neq j$. Fix $\mathbb{F}_q = \mathbb{F}_{p^e}$; let $s \geq 3$, and suppose there exists a $t$-bit, $s$-server QM scheme (Definition \ref{def: QM}) for $g_{i,j}$ and RS codes of dimension $k$ over $\mathbb{F}_q$. Then the download bandwidth satisfies 
    \begin{equation}\label{eqn: p case bound}
        t \geq \begin{cases}
            2\log_2(q-1) - 3 & p > 2\\
            2\log_2(q-2) - 4 & p = 2.
        \end{cases}
    \end{equation}
\end{restatable}

As noted above, this result shows that computing the quadratic monomial $g_{i,j} = f_i f_j$ requires download bandwidth nearly equal to the cost of representing both $f_i, f_j \in \mathbb{F}_q$ separately.  When  $k=2$, 
Theorem \ref{thm: full-case bound} shows that naive polynomial interpolation is essentially optimal for the problem of computing quadratic monomials. 

Additionally, we establish an even stronger lower bound for leakage functions that are linear over the base field $\mathbb{F}_p$.  This is notable because prior work on computing \emph{linear} functions $g$ often crucially relies on $\mathbb{F}_p$-linear leakage functions~\cite{GW17,TYB18,SW21,Kiah24}.\footnote{We note that the method used to aggregate the leakage functions may be non-linear, and so it is not trivial that the leakage functions must be non-linear to compute a non-linear function.}

\begin{theorem}[Informal; see Theorem \ref{thm: linear eval impossible}]\label{thm: linear eval impossible (informal)}
    Let $k \geq 2$ and $i,j \in \left[ 0, k-1 \right]$ $\mathbb{F}_q$ denote an arbitrary extension field of order $p^e$ where $e \geq 2$. Suppose each server is restricted to evaluating $\mathbb{F}_p$-linear functions $g_i:\mathbb{F}_q \to \mathbb{F}_p$ on their codeword symbols. Then there does not exist any $t$-bit, $s$-server QM scheme for $g_{i,j}$ and RS codes of dimension $k$ with download bandwidth $t$ satisfying $t < 2 \log_2(q)$ bits.
\end{theorem}

When $k=2$, Theorem~\ref{thm: linear eval impossible (informal)} implies that no scheme with linear leakage functions can perform \emph{even one bit} better than the naive strategy of recovering all of $\mathbf{f}$. 

\vspace{.5cm}
As mentioned above, we view our results as an important step towards understanding the bandwidth cost of general non-linear computation.  But the quadratic monomial case is already interesting in the context of prior work across many domains.  We discuss these connections more in  Section~\ref{sec: related works}.  Highlights include:
\begin{itemize}
    \item Combined with existing work on regenerating codes~\cite{TYB18,CT22}, our work implies that for $k=2$, there are many linear functions $f_1 \alpha + f_0$ that can be computed with substantially less bandwidth than the quadratic monomial $f_0 \cdot f_1$.
    \item Our work contrasts with work on \emph{low-bandwidth secret sharing}.  It is known to be possible to recover a secret shared with standard Shamir sharing with non-trivial bandwidth~\cite{Huang16,HB17}; our work implies that, for a natural multiplicative variant of Shamir sharing, no deterministic low-bandwidth recovery is possible.
  
   Unfortunately, this does not imply that ``multiplicative Shamir sharing'' is \emph{leakage resilient} in the information-theoretic sense, as we demonstrate concretely in Appendix~\ref{apx: LBR example}.
    \item Our work provides lower bounds on the download bandwidth single-client \emph{homomorphic secret sharing}~\cite{BGI16,BGILT18}, for multiplication of two secrets with a natural multi-secret version of Shamir sharing. 
\end{itemize}

\subsection{Technical Overview}\label{sec: technical overview}

We now overview the proof Theorem \ref{thm: full-case bound}, which we prove in Section~\ref{sec:mainpf}. First, we observe that it suffices to prove Theorem~\ref{thm: full-case bound} in the case where $k=2$ and $i=0$, $j = 1$. At a high level, this follows from the fact that for arbitrary $k$ and distinct $i,j \in [0,k-1]$, any $t$-bit, $s$-server QM scheme $\Phi$ computing $g_{i,j}(\mathbf{f})$ is a $t$-bit, $s$-server QM scheme over a two-dimensional subspace $\mathcal{C}$ containing all $\mathbf{f} = c_i e_i + c_j e_j$, where $c_i, c_j \in \mathbb{F}_q$, and $e_i, e_j \in \mathbb{F}_q^k$ are the standard basis vectors. If $\Phi$ can compute $g_{i,j}(\mathbf{f})$ with strictly fewer bits downloaded than required by Theorem \ref{thm: full-case bound}, then we may reduce it to a $k=2$, $i=0$, $j=1$ instance of the problem, noting that $\mathcal{C} \simeq \mathrm{RS}_q[n,2]$. For more detail, see Section~\ref{sec:mainpf}.

We thus assume $k=2$, $i=0$, and $j=1$ for the rest of this discussion. We begin with a sketch, before elaborating on each step. We first bound the bandwidth of any QM solution by the round complexity of an iterative algorithm, which partitions all the lines $f(x) = mx + b$ into \emph{buckets} determined by their coefficient product $g(\mathbf{f})\defeq g_{0,1}(\mathbf{f}) = mb$.  More precisely, for each $\gamma \in \mathbb{F}_q$, we partition the lines into buckets 

\begin{equation}\label{eqn: intro gamma bucket}
    B_\gamma \defeq \left\lbrace f(x) \in \mathbb{F}_q[x] \,:\, \deg(f) \leq 1, g(\mathbf{f}) = \gamma \right\rbrace.
\end{equation}

 We then imagine receiving bits from the leakage functions sequentially, one per round; each round, we ``prune'' away all of the lines $f(x)$ disagreeing with the leakage bit. The algorithm terminates when only one non-empty bucket $B_\gamma$ remains, corresponding to the correct coefficient product $\gamma$.  If an instance of QM has bandwidth $t$, then the corresponding instance of this algorithm must halt within $t$ steps.  Thus, the round complexity of the iterative algorithm yields a lower bound on the bandwidth $t$.

Analyzing how these buckets of lines evolve as we prune them seems challenging, so our second step restricts both the sets of possible servers and of possible coefficient products in order to introduce symmetry aiding our analysis. In more detail, we demand that the coefficient product $g(\mathbf{f})$ belongs to a specially structured set $\Omega_q \subseteq \mathbb{F}_q$ of size about $q/2$.  (When $q$ is odd $\Omega_q$ is the set of \emph{quadratic residues}, and when $q$ is even it is the even powers of a specially chosen primitive element.)  We similarly restrict queries to servers indexed by $\Omega_q$.  This assumption is without loss of generality: restricting the value of $g(\mathbf{f})$ makes the QM problem easier, and thus makes impossibility results stronger; furthermore, restricting the set of servers is allowed because in Definition~\ref{def: QM}, the client must be able to query \emph{any} set of $s$ servers.

The set $\Omega_q$ is designed so that restricting to this special case introduces useful symmetry in the buckets $B_\gamma$ described above.  In our third step, we use this symmetry to carefully ``project'' each bucket $B_\gamma \subseteq \mathbb{F}_q[x]$ of lines onto a subset of $\mathbb{F}_q$, in a way that the higher dimensional characteristics of lines are sufficiently represented in the lower dimensional projection. 
In particular, we show that when one runs an analogous algorithm on the projected buckets (iteratively pruning out \emph{projections} of lines that are inconsistent with the leaked bits), then it remains the case that the round complexity of the projected algorithm is a lower bound on the bandwidth of QM.

The fourth and final step is to bound the round complexity of the projected algorithm.  This problem turns out to be more tractable than the original problem of differentiating buckets of lines, and via the logic above, it implies Theorem \ref{thm: full-case bound}.

Next, we expand slightly upon each of these steps.

\paragraph{Step 1: Algorithmic View of QM.} 
Let $B_\gamma$ be given by Equation~\eqref{eqn: intro gamma bucket}, and suppose there exists a $t$-bit, $s$-server QM scheme for $g(\mathbf{f}) := g_{i,j}(\mathbf{f})$ as in Definition~\ref{def: QM}. For a fixed set $S$ of $s$ servers, denote the leakage functions of the QM by $\{\lambda_{i}\}$. Consider the following algorithm for recovering $g(\mathbf{f})$ given the leakage bits $\lambda_{i}(f(\alpha_i))$, $i \in [t]$.  

\vspace{.5cm}
\noindent \emph{Algorithmic view of QM:}
\begin{enumerate}
    \item Initialize a set $B_\gamma^0 \defeq B_\gamma$ for each $\gamma \in \mathbb{F}_q$.

    \item For each round $i = 1, 2, \ldots, t$:
    \begin{enumerate}

    \item Learn $\lambda_i(f(\alpha_i))$ by querying the server indexed by $\alpha_i$.

    \item For each $\gamma$, remove from $B_\gamma^{i-1}$ all $\tilde{f}(x) \in B_\gamma^{i-1}$ such that $\lambda_i(\tilde{f}(\alpha_i)) \neq \lambda_i(f(\alpha_i))$; call this pruned set $B_\gamma^i$.

    \item If there is only one $\gamma$ so that $B_\gamma^i$ that is non-empty, return.
    \end{enumerate}
\end{enumerate}
After $t$ rounds, the correctness of the QM implies that $B_\gamma^t = \varnothing$ for all but one value $\gamma = g(\mathbf{f})$.  In particular, the bandwidth $t$ of the QM is bounded below by the number of iterations that the algorithm above runs for before returning.

\paragraph{Step 2: Restricting to a special set.} We show that restricting both the coefficient products and the evaluation points to a subset $\Omega_q \subseteq \mathbb{F}_q^\ast$ of the multiplicative subgroup introduces exploitable symmetry in $B_\gamma$.
\begin{itemize}
    \item When field characteristic is $p > 2$, we let $\Omega_q$ be the quadratic residue subgroup.
    \item When field characteristic is $p=2$, we let $\Omega_q$ be the set $\lbrace \omega^{i} : i \text{ even} \rbrace$ for some primitive element $\omega \in \mathbb{F}_q^\ast$.
\end{itemize}
See Remark~\ref{rem:separate} below for more about why separate definitions are needed for $p$ even and $p$ odd.  Note that in both cases, $\Omega_q$ consists of quadratic residues;\footnote{Indeed, every element of a binary extension field is a quadratic residue.} in particular, for any $\alpha \in \Omega_q$, we may define an element $\sqrt{\alpha} \in \mathbb{F}_q$ such that $(\sqrt{\alpha})^2 = \alpha$. Observe that
\begin{equation*}
    B_1 = \left\lbrace m^{-1} x + m : m \in \mathbb{F}_q^\ast \right\rbrace.
\end{equation*}
Fix $\gamma \in \Omega_q$; one may then rewrite $B_\gamma$ as
\begin{equation*}
    B_\gamma = \left\lbrace \sqrt{\gamma} m^{-1} x + \sqrt{\gamma} m : m \in \mathbb{F}_q^\ast \right\rbrace \subseteq \mathbb{F}_q[x]
\end{equation*}
and observe that in fact $B_\gamma = \sqrt{\gamma}\cdot B_1$. This correspondence between sets of lines extends to a correspondence between their images under evaluation at a given point $\alpha \in \mathbb{F}_q^\ast$:
\begin{equation*}
    B_\gamma(\alpha) \defeq \left\lbrace \sqrt{\gamma} m^{-1}\alpha + \sqrt{\gamma} m \right\rbrace = \sqrt{\gamma} \cdot B_1(\alpha).
\end{equation*}
When we further restrict $\alpha$ to also be in $\Omega_q$, we show that the following symmetry holds, allowing QM to be simplified considerably.

\begin{lemma}[Informal; see Lemma \ref{lemma: bucket rescaling}]\label{lem: informal bucket rescaling}
    Let $\sqrt{\gamma}, \sqrt{\alpha} \in \mathbb{F}_q^\ast$ be field elements whose squares are $\gamma, \alpha$, respectively. Then $B_\gamma(\alpha) = (\sqrt{\gamma})(\sqrt{\alpha})\cdot B_1(1)$.
\end{lemma}

Given this symmetry, we consider the sub-case of QM wherein we are guaranteed that $g(\mathbf{f}) \in \mathbb{F}_q^\ast$ is in $\Omega_q$, and the user is restricted to querying servers holding $f(\alpha)$ where $\alpha \in \mathbb{F}_q^\ast$ is also  in $\Omega_q$. As noted previously, this restriction is without loss of generality for the purpose of proving a lower bound on the bandwidth of QMs.

\paragraph{Step 3: Reduction to distinguishing subsets of $\mathbb{F}_q$.} When $\gamma \in \Omega_q$, we show that the correspondence in Lemma \ref{lem: informal bucket rescaling} can be leveraged to ``project'' each $B_\gamma \subseteq \mathbb{F}_q[x]$ to a set of points $ S_\gamma \defeq \varepsilon_\gamma B_1(1) \subseteq \mathbb{F}_q$, where $\varepsilon_\gamma \in \mathbb{F}_q^\ast$. This reduces the problem of finding which $B_\gamma$ contains $f(x)$ to that of finding which $S_\gamma$ contains some distinguished point $\zeta \in \mathbb{F}_q$. 

In Definition~\ref{def: converted eliminators} and Theorem~\ref{thm: eliminator conversion}, we show explicitly how to map a leakage function $\lambda$ to a set $V \subseteq \mathbb{F}_q$ so that if the line $f(X)$ has $\lambda(f(\alpha_i)) = 1$, then a distinguished point $\zeta$ lies in $V$.  
Thus, given our QM with leakage functions $T_1, \ldots, T_t$, we obtain a list of corresponding sets $V_1, \ldots, V_t$ which can be used to iteratively prune the sets $S_\gamma$ until only one (say, $S_\delta$) remains nonempty. Assembling these insights yields the following algorithm, analogous to the one above, except that we now iteratively prune the projected sets $S_\gamma \subseteq \mathbb{F}_q$.

\vspace{.5cm}
\noindent \emph{Algorithmic view of the projected QM:}
\begin{enumerate}
    \item Initialize a set $S_\gamma^0 \defeq S_\gamma$ for each $\gamma \in \Omega_q$. 

    \item For each such round $i=1,2,\ldots, t$:
    \begin{enumerate}
        \item For each $\gamma$, prune $S_{\gamma}^{i-1}$ by the $V_i$ and call the pruned set $S_\gamma^i$; that is, $S_\gamma^i \defeq S_\gamma^{i-1} \setminus V_i$.
        
        \item If there is at most one $\gamma$ so that $S_\gamma^i$ is non-empty, return.
    \end{enumerate}
\end{enumerate}
We call this algorithm ``projection QM'', or pQM for short. The version here is simplified to convey the main gist; see Algorithm \ref{alg:pQM} and Section \ref{sec: bounding subcase} for a formal description.

We show that, as with our algorithm on buckets of lines, if the projected algorithm is allowed to run for all $t$ rounds, then the correctness of the original QM implies that there will be at most one $\delta$ so that $S_\delta^t$ is non-empty. This implies that the round complexity of this projected algorithm is again a lower bound on the bandwidth of the original QM. The round complexity of this projected algorithm is much easier to analyze, leading to our final step.

\paragraph{Step 4: Analyzing the projected algorithm.}
The key to our analysis is to show that, when the projected algorithm terminates, most elements of $S_\delta$ have been removed from consideration.  

\begin{theorem}[Informal; see Theorem \ref{thm: output list size 2}]\label{thm: informal output list size 2}
    Suppose that the projected QM algorithm above terminates after $\ell$ rounds. Then $|S^\ell_\delta| \leq 2$ if the field characteristic satisfies $p > 2$, and $|S^\ell_\delta| \leq 3$ if the field characteristic satisfies $p=2$.
\end{theorem}

The intuition is to observe that each set $S_\gamma^0$ has size $|S_\gamma^0| \approx q/2$ that is half the field, and there are $\approx q/2$ such sets, each indexed by some $\gamma \in \Omega_q$. As a result, the sets $S_\gamma^0$ must overlap with each other significantly. Any field element $\alpha \in S_\gamma^0$ is held by many other $S_{\gamma'}^0$, where $\gamma' \neq \gamma$ are distinct elements of $\Omega_q$. Hence, distinguishing a unique $S_\delta$ among these will require most elements of the entire field to have been removed from consideration.

With Theorem \ref{thm: informal output list size 2} established, the bound of Theorem \ref{thm: full-case bound} follows in Section \ref{subsubsec: adversarial pQM} by considering ``adversarial but honest'' servers who always reply to queries with the bit that prunes the fewest elements; that is, at most half of the set. Since the projected algorithm cannot terminate unless there are two or fewer elements remaining among all sets $S_\gamma^i$, we see that in the worst case, it takes about $\log_2(q^2)$ rounds for the algorithm to terminate. The exact expression seen in Theorem \ref{thm: full-case bound} follows from a more precise accounting.

\subsection{Related Work}\label{sec: related works}

\subsubsection{Computing linear functions: Regenerating codes and beyond} 

\paragraph{Regenerating Codes.}
Existing work has considered our model in the case when $g(\mathbf{f})$ is linear.  \emph{Regenerating codes} focus on a particular subset of linear functions.
Regenerating codes (e.g., \cite{DGW10,dimakis_survey}) are error correcting codes equipped with algorithms to efficiently repair a single erased codeword symbol in a distributed storage system.  In more detail, some data $\mathbf{f}$ is encoded as a codeword $\mathbf{c}$, and each symbol $c_i$ is sent to a different server.  If one server $i^*$ becomes unavailable, the goal is to compute $c_{i^*}$ using as little information as possible from a subset of $s$ surviving nodes.  Repair is a special case of our model where the function to compute is $g(\mathbf{f}) = c_{i^*}$, whose linearity follows from the linearity of the code.  A long line of work has established constructions of optimal regenerating codes in many parameter regimes; most relevant to our work is the study of Reed-Solomon codes as regenerating codes, initiated by \cite{Shanmugam14}.  By now, it is known that RS codes can be optimal or near-optimal regenerating codes in many parameter regimes.  For example, \cite{GW17,Dau21} show that full-length RS codes (that is, with $q=n$) of rate $1-\epsilon$ can achieve bandwidth bandwidth $(n-1)\log(1/\eps)$, and that this is nearly optimal for linear repair schemes.  For constant $\eps$, this is an $\Omega(\log n)$ improvement over the naive bound of $k \log q$.  When $q$ is much larger than $n$, it is possible to do better: Work by \cite{TYB18} provides repair schemes for RS codes achieving the \emph{cut-set bound}~\cite{DGW10}.  In our language, this gives bandwidth $s \log_2(q) / (s-k+1)$ bits, where $s$ is the number of surviving servers the scheme contacts.  As $s$ gets large relative to $k$, this can approach $\log_2 q$, the number of bits needed to write down $g(\mathbf{f})$.  In particular, this is also significantly less than the naive bound of $k \log_2 q$. 

All the constructions discussed above use linear repair schemes over extension fields $\mathbb{F}_{p^e}$, meaning that the local computation functions $g_i$, and the function used to aggregate them, are linear over $\mathbb{F}_p$.  The work \cite{CT22} studies the problem over prime order fields, where there are no non-trivial subfields and hence no linear repair schemes. 
For $k=2$, \cite{CT22} construct non-linear repair schemes which asymptotically converge to the cut-set bound over prime order fields $\mathbb{F}_p$ as $p \to \infty$.

\paragraph{Computing general linear functions.}
The work \cite{SW21} generalized the regenerating code model to consider the case when $g: \mathbb{F}_q^k \to \mathbb{F}_q$ is an \emph{arbitrary} linear function.  Given a length $n = q$ Reed-Solomon code of dimension $k = (1-\varepsilon)n$ over an extension field of constant characteristic, \cite{SW21} constructs a scheme which recovers $g(\mathbf{f})$ for any linear function given access to any $s = n(1-\gamma)$ servers, with download bandwidth $O(n/(\varepsilon-\gamma))$.  When $\eps$ is constant, this is a factor of $\Omega(\log n)$ improvement over the naive bound of $k \log q$ bits.

In follow-up work, \cite{Kiah24} considered reconstructing $\ell$-sparse linear combinations of codeword symbols, for $\ell \leq k$.  When $\mathbb{F}_q = \mathbb{F}_{p^e}$ has non-trivial extension degree $e\geq 2$ over $\mathbb{F}_p$, they give a low-bandwidth scheme for evaluating $g$, which downloads $d \log_2(p)$ bits, where $d \geq \ell p^{e-1} - \ell + k$. When $p,e$ are carefully chosen and $\ell \ll k$, their construction can outperform the naive $k e \log_2(p)$ download bandwidth.  

It is interesting to compare the results for computing linear functions to our results on computing quadratic monomials in the $k=2$ setting.  
We first note that we cannot directly compare our results to those of \cite{SW21,Kiah24}: 
While those works show that the naive bound of $k \log_2(q)$ can be significantly beaten for linear functions in some parameter regimes, the results are not meaningful for $k=2$.\footnote{In more detail, the focus of \cite{SW21} is on high-rate codes, while the results of \cite{Kiah24} require $k$ to be large enough that the sparsity $\ell$ can be much smaller than $k$.}   However, we can compare our results to the results of \cite{TYB18} and \cite{CT22}. Given a message polynomial $f(x) = mx + b$, these papers together show that it is possible to compute certain\footnote{%
For $k=2$, it may look like the problem of computing $g(\mathbf{f}) = f(\alpha)$ is the same as computing arbitrary linear functions, as all linear functions (up to normalization) of $(b,m)$ look like $b + \alpha m$ for some $\alpha$.  However, these are not the same problem when the set of evaluation points for the RS code is not the entire field.  The regenerating code constructions in \cite{TYB18,CT22} have $n \ll q$, so those works do not immediately give a scheme for computing general linear functions, even when $k=2$.}  
linear functions of the form $m \alpha + b$ with bandwidth approaching $\log_2 q$ as the number of contacted servers $s \to \infty$ is sufficiently large and $|\mathbb{F}| \to \infty$, over both prime fields $\mathbb{F}_p$ \cite{CT22}; and over extension fields $\mathbb{F}_{p^e}$ for suitably large $e$ \cite{TYB18}. Our result works for both prime order fields and extension fields, and when $s$ is arbitrarily large (noting that without loss of generality we have $ s\leq k \log_2 q$). Thus, in parameter regimes where these results overlap, we see that recovering the linear function $m\alpha + b$, for any evaluation point $\alpha$ in the RS codes considered by \cite{CT22} or \cite{TYB18}, requires substantially less bandwidth than recovering the quadratic function $m \cdot b$.

\subsubsection{Secret Sharing}
In a secret-sharing scheme, a secret $s \in \mathbb{F}_q$ is shared among $n$ parties.  One goal is that any $k$ of the parties should be able to combine their shares to recover the secret, while any $k-1$ parties together learn nothing about the secret.  The secret sharing scheme most relevant to our work is \emph{Shamir's scheme}~\cite{Shamir79}.  In Shamir's scheme, secret $s \in \mathbb{F}_q$ is shared by sampling $f_i \gets \mathbb{F}_q$, $i \in \{1,\ldots, k-1\}$, uniformly at random and considering $f(x) = s + \sum_{i=1}^{k-1} f_i x^i$. 
The parties are indexed by elements $\alpha \in \mathbb{F}_q^*$, and their shares are given by $f(\alpha)$.  That is, each party holds a symbol of a RS codeword.  It is not hard to see that this scheme has the desired access and security requirements.

There are several questions in secret sharing relevant to our work.  We discuss them below.

\paragraph{Low-Bandwidth Secret Sharing.} 
In \emph{low-bandwidth secret sharing}~\cite{Wang2008OnSR, ZHANG2012106,Huang16, HB17, Rawat18}, the goal is for any subset of enough parties to be able to reconstruct the secret in a communication-efficient way.  When the secret sharing scheme is Shamir, this is again a special case of our problem, where $g(\mathbf{f})$ is the function $g_0(\mathbf{f}) = f(0)$.  (The difference between this setting and regenerating codes is that we only need to be able recover $f(0)$, rather than $f(\alpha)$ for any $\alpha$).  In this setting, \cite{Huang16} shows that the cut-set bound is still the limit on the bandwidth (for any secret sharing scheme); but it is possible to attain this with a smaller alphabet size than needed for regenerating codes~\cite{HB17}.

Our work on the problem of QM may be viewed in the context of low-bandwidth secret sharing by considering a ``multiplicative'' variant of Shamir secret sharing.    In more detail, in the standard Shamir sharing with $k=2$, a slope $m$ is drawn at random and the shares correspond to the function $f(x) = mx + s$, where $s$ is the secret.  In the multiplicative variant, $m$ and $b$ are chosen at random so that the secret is $s = m \cdot b$.  This is a special case of QM, when $g_{0,1}(\mathbf{f})$ is restricted to be nonzero.  {In particular, our results imply that there do \emph{not} exist non-trivially low-bandwidth secret sharing schemes for ``multiplicative Shamir sharing'' with $k=2$, even though such schemes do exist for standard Shamir sharing.}

\paragraph{Leakage-Resilient Secret Sharing.}

The work \cite{BDIR18} considers a threat model in which an adversary has ``local leakage access'' to more than $k$ shares.  In this model, an adversary can apply an arbitrary function of bounded output length to each share locally. Concretely, we may think of this model as an adversary extracting a few bits of \textit{local} information from each shareholder. Under this view, for Shamir sharing, the adversary may be considered as a repair scheme wishing to recover the codeword symbol $f(0)$. A code is \textit{local leakage resilient} if the adversary has negligible advantage learning the secret $s$.

For sufficiently large code length $n$ and prime field order $p$, \cite{BDIR18} show that Reed-Solomon codes over $\mathbb{F}_p$ are local leakage resilient when dimension $k \geq \beta n$ for some constant $\beta$. As a concrete example, they show that if the adversary is allowed to leak 1 bit from each share, then for $n$ sufficiently large, it suffices to take $\beta \geq 0.92$. Extensive follow-up work (e.g., \cite{MPSW21, MNPW22, KK23, Kasser25}) has progressively lowered the threshold for 1-bit leakage resilience, with \cite{Kasser25} most recently improving the bound to $\beta \geq 0.668$. On the other hand, \cite{CT22, CSTW24} show that for low-degree RS codes (with $k =o(n)$), non-trivial leakage is possible over prime fields.

Given the discussion above about the ``multiplicative'' version of Shamir sharing, one may hope that our results imply that multiplicative Shamir sharing is leakage resilient.  Unfortunately, this seems not to be the case for dimension $k=2$ Reed-Solomon codes.  
 First, at least over some fields, it's possible to learn the entire secret from strictly less than $k \lceil\log_2(q-1)\rceil$ leaked bits, so there is some amount of non-trivial leakage that completely reveals the secret.\footnote{In this case, since the secret must be non-zero, the naive bound is $k \lceil\log_2(q-1)\rceil$ rather than $\lceil k \log_2(q) \rceil$.} As an example, we show in Appendix \ref{apx: LBR example} that ``multiplicative Shamir'' instantiated over $\mathbb{F}_7^\ast$ admits a reconstruction algorithm that downloads only $5$ bits, as compared to the $2\lceil \log_2(6) \rceil = 6$ bits one might expect in the naive case. Second, with this scheme, \textit{leaking even one bit} from a server allows an adversary to learn information about the multiplicative Shamir secret; we give an example of this as well in Appendix~\ref{apx: LBR example}.

\paragraph{Homomorphic Secret Sharing Schemes.} In (single-client) \emph{homomorphic secret sharing} (HSS)~\cite{C:Benaloh86a,BGI16a,BGILT18}, a secret $s$ is shared as above, and a referee subsequently wishes to compute a function $g(s)$ of the secret.  Each party is allowed to do some local computation and send a message to the referee.  In some applications, it is desirable to reduce the download bandwidth of the scheme. For example, in applications of HSS to Private Information Retrieval (PIR), the download bandwidth corresponds to the download cost of the PIR scheme (see~\cite{BCGIO17, FIKW22}).  Our problem of low-bandwidth function evaluation is related to HSS where we want to minimize the download bandwidth, and where we want \emph{information-theoretic} security.\footnote{We note that for some applications of HSS, the reconstruction should be additive; and/or the messages sent by the parties don't leak any information about the secret beyond $g(s)$.  Neither of these are necessarily the case in low-bandwidth function evaluation.}  

The work~\cite{FIKW22} gives \emph{multi}-client HSS schemes for Shamir sharing, where the referee's function $g$ is a low-degree polynomial.  At first glance this seems at odds with our result (which says that computing low-degree polynomials on top of Shamir sharing with non-trivial bandwidth is impossible).  However, the model for multi-client HSS is different than the one we consider, as the function $g$ is evaluated on multiple secrets, which are independently secret-shared.  For example, \cite{FIKW22} applies to a setting where a user wishing to compute the monomial $g(x,y) = xy$ on inputs $x = s_1$ and $y=s_2$ assumes that $s_1, s_2$ are shared separately with lines $f_1(x) = s_1 + m_1 x$ and $f_2(x) = s_2 + m_2 x$; this task is distinct from computing $s_i \cdot m_i$.

Our results for quadratic monomials \emph{do} have implications for low-bandwidth, \emph{single}-client secret sharing. For example, consider the secret sharing scheme that shares a secret $(a,b) \in \mathbb{F}_q^2$ via evaluations $f(\alpha)$ of a polynomial $f(x) = a + bx + \sum_{j=2}^{k-1} c_j x^j$, where the $c_i$ are chosen randomly.  This is a natural extension of Shamir's scheme to multiple secrets.\footnote{It is related to, but not the same as, the scheme in \cite{FY92}; the \cite{FY92} scheme can be seen as a systematic version of the scheme described above.} Our work implies that computing $a \cdot b$ requires a download bandwidth of nearly $2 \log q$ bits, twice as much as we might hope for.

\subsubsection{Coded Computation}
In \emph{coded computation} (e.g.,~\cite{LLPPR18, DCG19, YLRKSA19} or see \cite{CCsurvey} for a survey), the goal is to compute a function $g$ of some data $\mathbf{f}$, distributed among $n$ worker nodes.  The concern is that some worker nodes may unpredictable be \emph{stragglers} (slow or non-responsive), and we would like to carry on the computation without them.  The 
idea is to encode $\mathbf{f}$ as a codeword $\vec{c}$, so that $g(\mathbf{f})$ can be computed even if a the computation on a few symbols of $\vec{c}$ are unavailable.

This is similar in spirit to our model, but there are a few differences.  First, coded computation is often studied over $\mathbb{R}$, rather than finite fields---it is an interesting question whether a version of our result holds over $\mathbb{R}$.  Second, in our model the leakage functions are allowed to depend on the set $S$ of queried servers, which is not generally the case in coded computation.  However, we note that our lower bound would apply to coded computation leakage function model as well, as the problem of QM is harder if the leakage functions cannot depend on $S$.\footnote{A third difference is that, to the best of our knowledge, coded computation schemes that are not repetition-based do not attempt to minimize the download bandwidth, but rather the number of workers that need to finish before the function $g(\mathbf{f})$ can be computed, and/or an analysis of how much this speeds up the computation in the presence of stragglers.  However, minimizing bandwidth---or proving that it cannot be done---seems well-motivated.} Thus, while coded computation is similar in spirit to our model, extensions to our work (to the reals and to computations larger than the product of two field elements) would be needed to give meaningful bounds in this setting.

\subsection{Open Questions}
We mention a few open questions raised by our results. 

\begin{itemize}

\item It would be interesting to extend our results to higher degree monomials. Our results show that about $2\log_2(q)$ bits are necessary for computing quadratic monomials $g$ on RS-encoded data.  For monomials of degree $d > 2$, one conjecture is that about $d \log_2(q)$ bits are needed.  Is this conjecture true?

\item 
It would also be interesting to know if our lower bound is achievable.  That is, is it possible to evaluate quadratic monomials on top of Reed-Solomon codes of dimension $k > 2$, with bandwidth approaching $2 \log_2 q$?  For $k=2$, this can be done by naive polynomial evaluation. 
For $k > 2$, there are RS codes over $\mathbb{F}_q$ that admit repair of individual \emph{codeword} symbols with download bandwidth converging to $\log_2(q)$ bits~\cite{TYB17}.  If a similar result holds for \emph{message} symbols, this would provide an algorithm for evaluating quadratic monomials with bandwidth approaching our lower bound.  We are not aware of such a result (for message symbols) in the regenerating codes literature; does such a result hold, or are there other algorithms for computing quadratic monomials with about $2 \log_2 q$ bits?

\item Finally, it would be interesting to extend our results to arbitrary quadratic \emph{functions}, not just quadratic monomials.  It is not hard to see that at least $k$ bits are required to recover an arbitrary quadratic function (even an arbitrary \emph{linear} function~\cite{SW21}).  When $k$ is large, this implies that a stronger bound should hold; what is the correct bound?

\end{itemize}

\section{Preliminaries}

We now set notation and establish an algorithmic view of QM. Denote by $\mathbb{F}_q$ the finite field of $q = p^e$ elements, where $p$ is prime. For any positive integer $n \in \mathbb{Z}^+$, we denote by $[n]$ the sequence of integers $(1, 2, 3, \ldots, n)$. Given integers $i < j$, we denote by $[i,j] \subseteq \mathbb{Z}$ the sequence of integers $(i, i + 1, \ldots, j)$. For any $\mathcal{F} \subseteq \mathbb{F}_q[x]_{\deg \leq k-1}$, we may interpret any $f \in \mathcal{F}$ as a vector of coefficients  $\mathbf{f} = (f_i)_{i \in [0,k-1]} \in \mathbb{F}_q^{k}$: $f(x) = \sum_{i=0}^{k-1} f_i x^i$. For any $i, j \in [0,k-1]$, we write $g_{i,j}(\mathbf{f}) \defeq f_i f_j$.

 Our analysis throughout primarily focuses on the problem of finding the products of coefficients of linear polynomials; explicitly, this is the setting where $k=2$ and $i=0$, $j = 1$. 
We then generalize this to the full statement of Theorem~\ref{thm: full-case bound} in Section~\ref{sec:mainpf}.  Thus, until then we make the following assumption:
\begin{assumption}\label{assm:k2}
    Until Section~\ref{sec:mainpf}, we assume that $k=2$, so $f(x) = mx + b$, and that $g(\mathbf{f}) = g_{0,1}(\mathbf{f}) = m \cdot b$.
\end{assumption}

\subsection{Algorithmic View}\label{sec: algorithmic view}

We now give an algorithmic characterization of QMs (Definition \ref{def: QM}), in light of Assumption~\ref{assm:k2}.
Recall that the goal of QM is not to completely recover $f(x) = mx + b$, rather to learn the product of its coefficients $\gamma \defeq g(\mathbf{f}) = mb$. At least $q-1$ lines $f \in \mathbb{F}_q[x]$ have the same coefficient product $\gamma$.  As we do not need to distinguish between these lines, we put them in the same \emph{bucket}, defined below.

\begin{definition}[$\gamma$-Bucket]\label{def: gamma bucket}
    Given some $\gamma \in \mathbb{F}_q$, we define a $\gamma$-bucket as
    \begin{equation}
    B_\gamma = \left\lbrace f \in \mathbb{F}_q[x] \; : \; \deg(f) \leq 1 \text{ and } g (\mathbf{f}) = \gamma \right\rbrace \subseteq \mathbb{F}_q[x]
\end{equation}
\end{definition}

Thus, a $\gamma$-bucket is the set of all lines $f(x)$ whose coefficient product is $\gamma$. 

\begin{definition}[$\gamma$-Bucket Evaluation]\label{def: gamma bucket eval}
    Given $\alpha, \gamma \in \mathbb{F}_q$, we denote by
    \begin{equation}
    B_\gamma(\alpha) \defeq \left\lbrace f(\alpha) : f \in B_\gamma \right\rbrace \subseteq \mathbb{F}_q
\end{equation}
the set of all evaluations $f(\alpha)$ for $f \in B_\gamma$. 
\end{definition}

In other words, $B_\gamma(\alpha)$ is the image of the evaluation map $\mathrm{ev}_\alpha: \mathbb{F}_q[x] \to \mathbb{F}_q$, $f \mapsto f(\alpha)$ on a set of lines $B_\gamma$. We now show that any $t$-bit, $s$-server QM is equivalent to an instance of Algorithm \ref{alg: QM Full}. More precisely, an QM is valid if and only if the corresponding instance of Algorithm \ref{alg: QM Full} successfully recovers $g(\mathbf{f})$ for any line $f(x)$, given a transcript of leakage bits $\mathbf{b} \in \lbrace 0 ,1 \rbrace^t$ corresponding to $f(x)$. We state this formally in Observation \ref{obs: QM algorithmic equiv}.

\begin{algorithm}
\caption{$\mathrm{QM}(\pmb{\alpha}, \mathbf{T}, \mathbf{b})$}\label{alg: QM Full}
\DontPrintSemicolon

\SetKwInOut{Input}{Input}\SetKwInOut{Output}{Output}
\Input{server indices $S \subseteq \mathbb{F}_q$, $|S| = s$}
\Input{server schedule $\pmb{\alpha} = (\alpha_1, \alpha_2, \ldots, \alpha_{t}) \in S^t$} 
\Input{leakage functions $\mathbf{T} = (T_1, T_2, \ldots, T_t)$, $T_i \subseteq \mathbb{F}_q$}
\Input{leakage transcript $\mathbf{b} \in \lbrace 0, 1 \rbrace^{t}$}

\For{$\gamma \in \mathbb{F}_q$}
{
$\mathbf{B}(\gamma) \gets B_\gamma$ \\
\For{$i \in [t]$}
{
$\tilde{T} \gets \begin{cases} T_i & b_i = 0 \\ \mathbb{F}_q \setminus T_i & b_i = 1\end{cases}$ \\
    $\mathbf{B}(\gamma) \gets \mathbf{B}(\gamma) \setminus \{ f \in \mathbf{B}(\gamma) \,:\, f(\alpha_i) \in \tilde{T} \}$
}
}
\If{ $\exists! \; \gamma \in \mathbb{F}_q \;\ni\; \mathbf{B}(\gamma) \neq \emptyset$}
{
\Return{$\gamma$ (\texttt{Success!})}
}
\If{$\mathbf{B}(\gamma) =\emptyset \; \forall \; \gamma \in \mathbb{F}_q$}
{
\Return{\texttt{Invalid transcript!}}
}
\Return{\texttt{Fail!}}\\

\end{algorithm}

\begin{observation}
\label{obs: QM algorithmic equiv}
    There exists a $t$-bit, $s$-server QM for $g_{0,1}$ and for $k=2$ (Definition \ref{def: QM}) if and only if, for every choice of $S \subseteq \mathbb{F}_q$, $|S| = s$, there exists some $\pmb{\alpha} = (\alpha_1,\ldots, \alpha_t)  \in S^t$ and a collection of subsets $T_i \subseteq \mathbb{F}_q$, $i \in [t]$ defining
    \begin{equation}
        \lambda_i: \mathbb{F}_q \to \lbrace 0, 1 \rbrace, \; x \mapsto \begin{cases}
            0 & x \in T_i\\
            1 & \text{else}
        \end{cases}
    \end{equation}
such that Algorithm \ref{alg: QM Full} succeeds and outputs $g(\mathbf{f})$, given any input
    \begin{equation}
        \mathbf{b} = \left( \lambda_i(f(u_i)) : i \in [t] \right) \in \lbrace 0, 1 \rbrace^t
    \end{equation}
    for all $f \in \mathbb{F}_q[x]$, $\deg(f) \leq 1$.
\end{observation}

\begin{proof}
    The reverse direction is true by definition, so it suffices to consider the forward direction.  Suppose there exists a $t$-bit, $s$-server QM for $g_{0,1}$ and $k=2$. Fix a set $S \subseteq \mathbb{F}_q$ of $s$ servers.  Let $\lambda_1, \ldots, \lambda_t$ be the leakage functions guaranteed by Definition~\ref{def: QM}. Define
    \[ T_i = \lambda_i^{-1}(0) \subseteq \mathbb{F}_q\]
    for all $i \in [t]$, and let $\mathbf{T} = (T_1, \ldots, T_t).$ For an arbitrary $f(x) \in \mathbb{F}_q[x]$ of degree at most $1$, set
    \[ \mathbf{b}(f) = (\lambda_i(f(\alpha_i)) : i \in [t] ).\]
    Let $\pmb{\alpha}, \mathbf{T}, \mathbf{b}(f)$ be the inputs to Algorithm~\ref{alg: QM Full}. 
    The algorithm initializes a dictionary, denoted $\mathbf{B}$, whose value at key $\gamma \in \mathbb{F}_q$ is the $\gamma$-bucket $B_\gamma$ (Definition \ref{def: gamma bucket}). In each iteration $i \in [t]$ of the main loop, the algorithm considers the $i$th bit of $\mathbf{b}(f)$, which denotes whether $f(\alpha_i) \in T_i$ or $f(\alpha_i) \in \mathbb{F}_q\setminus T_i$; the algorithm removes all lines $g$ that are inconsistent with this leakage bit. 
    
    Since we assumed a QM exists, by Definition \ref{def: QM}, all $h\in \mathbb{F}_q[x]_{\deg\leq 1}$ consistent with $\mathbf{b}(f)$ satisfy $g(\mathbf{h}) = g(\mathbf{f}) = \gamma$, which implies that all lines $h$ that remain in some bucket must in fact all reside in the unique bucket $\mathbf{B}(\gamma)$. This guarantees that Algorithm 1 will terminate and correctly output $\gamma = g(\mathbf{f})$, as desired.
\end{proof}

Given the equivalence established by Observation \ref{obs: QM algorithmic equiv}, the minimal download bandwidth incurred by any QM is equivalent to the minimal round complexity of Algorithm \ref{alg: QM Full}, subject to the constraint that Algorithm \ref{alg: QM Full} outputs the correct coefficient product given any leakage transcript.

\subsection{Restricted Parameter Regime for QM}

In Section \ref{sec: algorithmic view}, we established that determining the optimal download bandwidth incurred by an instance of QM is equivalent to bounding the round complexity of the corresponding instance of Algorithm \ref{alg: QM Full}. Rather than directly analyzing Algorithm \ref{alg: QM Full}, we pursue an alternate strategy: within the scope of all parameter regimes for which QM must succeed, we find a subset whose corresponding algorithmic view can be greatly simplified.

Algorithm \ref{alg: QM Full} terminates when only one $\mathbf{B}(\gamma) \subseteq B_\gamma$ is non-empty. Assuming a valid QM, the algorithm succeeds because all lines consistent with a given leakage transcript have the same coefficient product. Instead of tracking the number of lines in each set $\mathbf{B}(\gamma)$, consider an alternate implementation of Algorithm~\ref{alg: QM Full} that tracks the size of $\mathrm{ev}_\alpha(\mathbf{B}(\gamma)) = \{ f(\alpha) \,:\, f \in \mathbf{B}(\gamma) \} \subseteq \mathbb{F}_q$ for any $\alpha \in \mathbb{F}_q^\ast$.  This view is equivalent, since for any such $\alpha$,
\[ \mathrm{ev}_\alpha (\mathbf{B}(\gamma)) = \emptyset\   \iff \mathbf{B}(\gamma) = \emptyset.\]
It turns out that we can construct such a correspondence between linear functions and their evaluations which, within a restricted parameter regime, yields a simpler but equivalent algorithm. To state this restricted parameter regime, we first give the following definition.

\begin{definition}[Quadratic Residues]
    Given a field $\mathbb{F}_q$, let
    \begin{equation}
        \mathrm{QR}_q \defeq \left\lbrace \alpha^2 : \alpha \in \mathbb{F}_q^\ast \right\rbrace
    \end{equation}
    denote the quadratic residue subgroup of $\mathbb{F}_q^\ast$.
\end{definition}

When $q=p^e$ where $p>2$, exactly half of $\mathbb{F}_q^\ast$ lies in $\mathrm{QR}_q$, while the other half lies in $\mathbb{F}_q \setminus \mathrm{QR}_q$.  We state this formally below.

\begin{claim}[e.g., \cite{ethz_quadratic_residues}]\label{claim: QRq is half of Fq star}
    If $\mathbb{F}_q$ is a field with characteristic $p > 2$, then $|\QR_q| = (q-1)/2$. 
\end{claim}

Note that $\QR_q$ is the image of $\mathbb{F}_q^\ast$ under the map $x \mapsto x^2$; when $q = 2^e$, the map $x \mapsto x^2$ is the Frobenius automorphism, implying $\QR_{2^e} = \mathbb{F}_{2^e}^\ast$. Since the quadratic residues do not yield an interesting subset of the multiplicative subgroup in binary extension fields, we consider an alternative construction. First, we recall the definition of the field trace.

\begin{definition}
    For prime $p$ and $q=p^e$, the field trace $\mathrm{Tr}:\mathbb{F}_q \to \mathbb{F}_p$ is the $\mathbb{F}_p$-linear function given by
    \begin{equation*}
        \mathrm{Tr}(x) = x + x^{p} + x^{p^2} + \cdots + x^{p^{e-1}} = \sum_{i=0}^{e-1} x^{p^i}.
    \end{equation*}
\end{definition}

The following result\footnote{This result holds in considerably more generality: given any finite extension field $E$ and a field trace from $E$ to a subfield $F$ - not necessarily the base field - there exists a primitive element $\omega \in E^\ast$ whose trace $\mathrm{Tr}(\omega) = \mathrm{Tr}(\omega^{-1}) = \alpha \in F$ for arbitrary $\alpha \in F$. Note that $\omega$ is primitive if and only if $\omega^{-1}$ is primitive. See \cite{LN97}, Theorem 3.75.} (\cite{Cohen1990, LN97}) shows that, in all binary extension fields of order at least $2^e \geq 8$, there always exists a generator of the multiplicative subgroup whose inverse lies in the kernel of the field trace.

\begin{theorem}[\cite{Cohen1990, LN97}]\label{thm: prescribed trace}
    For any integer $e \geq 3$, there exists a primitive element $\omega \in \mathbb{F}_{2^e}^\ast$ such that $\mathrm{Tr}(1/\omega) = 0$.
\end{theorem}

\begin{definition}\label{def: pseudo residues over binary extensions}
    Let $e\geq 3$ and $q = 2^e$; given a primitive element $\omega \in \mathbb{F}_{q}^\ast$ satisfying $\mathrm{Tr}(1/\omega) = 0$, we denote 
    \begin{equation*}
        W_q(\omega) \defeq \left\lbrace \omega^{2i} : i = 0, 1, \ldots, 2^{e-1} - 2 \right\rbrace
    \end{equation*}
    Since such a primitive element $\omega$ is guaranteed to exist by Theorem \ref{thm: prescribed trace}, we will assume that there is a canonical choice of $\omega$ for each field order $q = 2^e$; we may thus drop the argument and simply write $W_q$.
\end{definition}

We may now state our restricted parameter regime.

\begin{definition}[Restricted Parameter Regime]\label{def: restricted parameter regime}
    Fix a QM for $g_{0,1}$ and $k=2$ as in Definition \ref{def: QM}. We consider the following restricted parameter regime. Let $q = p^e$ for prime $p$ satisfying $p>2$; or $p=2$ and $e \geq 3$. Define
    \begin{equation*}
        \Omega_q \defeq \begin{cases}
            \mathrm{QR}_q & p > 2\\
            W_q & p = 2, e \geq 3.
        \end{cases}
    \end{equation*}
\begin{enumerate}
    \item We only consider lines $f \in \mathbb{F}_q[x]$ such that $g(\mathbf{f}) \in \Omega_q$, and
    \item during any round $i \in [t]$, we only contact servers indexed by some $\alpha_i \in \Omega_q$.
\end{enumerate}
\end{definition}

As discussed in the introduction and formalized in Section~\ref{sec:mainpf}, restricting to this parameter regime is without loss of generality.  That is, any impossibility result for the restricted parameter regime in Definition~\ref{def: restricted parameter regime} implies an impossibility result for QM in general. 

\begin{remark}[Why two separate definitions of $\Omega_q$?]\label{rem:separate}
    Observe that for finite fields $\mathbb{F}_q = \mathbb{F}_{p^e}$ of odd prime characteristic $p$, the quadratic residue subgroup $\mathrm{QR}_q$ is precisely the set of even powers of any primitive element $\omega \in \mathbb{F}_{q}^\ast$. Thus, the construction of $\Omega_q$ in Definition \ref{def: pseudo residues over binary extensions} when $q = 2^e$ extends this notion to to the binary characteristic case, considering all even powers of a primitive element of the multiplicative subgroup. A natural question is why we did not simply construct $\Omega_q$ as the even powers of a primitive element for all field characteristics; the answer is that our analysis in the binary characteristic case will make extensive use of the fact that the primitive element $\omega \in \mathbb{F}_{2^e}^\ast$ satisfies $\mathrm{Tr}(\omega^{-1}) = 0$. Such primitive elements also exist in odd characteristic fields, but require extension degree at least 3 (analogous to Theorem \ref{thm: prescribed trace}, which in fact holds for odd characteristic). This would cause our analysis to overlook prime-order fields, as well their extensions of degree 2.
\end{remark}

\subsubsection{Odd Field Character}

When $q = p^e$ for odd $p$, the restriction to quadratic residues in Definition \ref{def: restricted parameter regime} allows us to leverage results concerning quadratic character sums in Section \ref{sec: eval map images}. We first define the notion of a quadratic character\footnote{The quadratic character is also commonly referred to as the Legendre symbol.} and note a pair of elementary properties.

\begin{definition}[Quadratic Character]\label{def: quadratic char}
    For $\mathbb{F}_q = \mathbb{F}_{p^e}$ for prime $p>2$, we define the quadratic character by
    \begin{equation}
        \chi: \mathbb{F}_q \to \left\lbrace0, \pm 1 \right\rbrace, \; x \mapsto \begin{cases}
            1 & x \in \mathrm{QR}_q\\
            0 & x = 0\\
            -1 & x \in \mathbb{F}_q^\ast \setminus \mathrm{QR}_q.
\end{cases}
    \end{equation}
\end{definition}

\begin{observation}[e.g., \cite{babaiWeilNotes}]\label{obs: two quadratic char props}
    Let $\mathbb{F}_q = \mathbb{F}_{p^e}$ for prime $p>2$, and let $\chi: \mathbb{F}_q \to \lbrace 0, \pm 1 \rbrace$ be the quadratic character given by Definition \ref{def: quadratic char}. 
    \begin{enumerate}
        \item For all $x \in \mathbb{F}_q^\ast$, $\chi(x) = \chi(1/x)$.
        \item For all $x,y \in \mathbb{F}_q$, $\chi(xy) = \chi(x) \chi(y)$.
    \end{enumerate}
\end{observation}

Observe that the quadratic characters $\chi(x)$ for all $x \in \mathbb{F}_q$ sum to zero, as $\mathrm{QR}_q$ contains precisely half the elements of $\mathbb{F}_q^\ast$. We may consider a more general sum of quadratic characters, which asks for $\sum_{x \in \mathbb{F}_q}\chi(f(x))$ for some $f \in \mathbb{F}_q[x]$. 

\begin{definition}[Complete Quadratic Character Sum]\label{def: complete quad char sum}
    Let $\mathbb{F}_q = \mathbb{F}_{p^e}$ for prime $p>2$, and let $\chi: \mathbb{F}_q \to \lbrace 0, \pm 1 \rbrace$ be the quadratic character given by Definition \ref{def: quadratic char}. For any $f \in \mathbb{F}_q[x]$, we define the complete quadratic character sum of $f$ by
    \begin{equation}
        \mathfrak{C}(f) \defeq \sum_{x \in \mathbb{F}_q} \chi(f(x)).
    \end{equation}
\end{definition}

Intuitively, the complete quadratic character $\mathfrak{C}(f)$ of any polynomial $f \in \mathbb{F}_q[x]$ is an integer which indicates whether the image of $f$ contains more quadratic residues or non-residues.

\begin{theorem}[Weil's Complete Character Sum Estimate\footnote{This bound was proven in considerably more generality by Andr\'e Weil in 1948 \cite{weil1948courbes}; modern proofs may be found in \cite{HartshorneAG,iwaniec2004analytic}. The version presented by Theorem \ref{thm: weil char sum estimate} is simplified for our setting, and has been adapted from L\'aszl\'o Babai's excellent lecture notes on the topic \cite{babaiWeilNotes}.} \cite{weil1948courbes,babaiWeilNotes}]\label{thm: weil char sum estimate}
    Let $\mathbb{F}_q = \mathbb{F}_{p^e}$ where prime $p>2$, and let $f \in \mathbb{F}_q[x]$ be a square-free polynomial of degree $d \defeq \deg(f)$. Then
    \begin{equation}
        \left\lvert \sum_{x \in \mathbb{F}_q} \chi\left( f(x) \right) \right\rvert \leq (d-1) \sqrt{q}.
    \end{equation}
\end{theorem}

\subsubsection{Even Field Character}

We leverage a different result to later analyze the case where $q = 2^e$. The following theorem is classical (e.g., \cite{lang2002algebra}); we narrowly restate its claim in our parameter regime of interest and include a proof for completeness.

\begin{theorem}\label{thm: binary artin schreier}
    Let $c \in \mathbb{F}_{2^e}$. Then there exists $y \in \mathbb{F}_{2^e}$ such that $ y^2 + y + c = 0$ if and only if $\mathrm{Tr}(c) = 0$.
\end{theorem}

\begin{proof}
    Let\footnote{Theorem~\ref{thm: binary artin schreier} can be seen as an additive form of Hilbert's 90th Theorem \cite{Hilbert1897}. The notation in our proof pays homage to Emil \textbf{A}rtin and Otto \textbf{S}chreier, who also proved a far more general version of this theorem in \cite{artin_schreier1927}.} $S(X) \defeq X^2 + X$ and $A(X) \defeq S(X) + c$. In the forward direction, suppose $A(y) = y^2 + y + c = 0$ and take the trace of both sides to observe $\mathrm{Tr}(y^2 + y) + \mathrm{Tr}(c) = 0$. Since $\mathrm{Tr}(y^2 + y) = 0$, it follows that $\mathrm{Tr}(c) = 0$. In the other direction, suppose that $c \in \ker(\mathrm{Tr})$. It suffices to show that the image of $S(X)$ is precisely the kernel of the trace map.
    
    Recall that $\mathrm{Tr}$ is an $\mathbb{F}_2$-linear map whose kernel has dimension $e-1$. Observe that $S(X) = X(X+1)$ is also an $\mathbb{F}_2$-linear map with kernel precisely $\mathbb{F}_2$; by rank-nullity, its image $\mathrm{img}(S)$ also has dimension $e-1$. Since $\mathrm{Tr}(S(X)) = 0$ for all $X \in \mathbb{F}_{2^e}$, we have $\mathrm{img}(S) \subseteq \ker(\mathrm{Tr})$; since they have the same dimension over $\mathbb{F}_2$, we see that $\mathrm{img}(S) = \ker(\mathrm{Tr})$, as desired.
\end{proof}

\section{Symmetries of restricted parameter regime}\label{sec: restriction implications}
The primary goal of this section is to apply the restricted parameter regime (Definition \ref{def: restricted parameter regime}) towards constructing a simplified version of Algorithm \ref{alg: QM Full} whose round complexity will under-bound that of the general case. The actual construction of this simplified version, along with its analysis, is given in Section \ref{sec: bounding subcase}. Here, we focus instead on demonstrating the simplifications afforded by the restricted parameters and on developing the techniques to re-cast QM into this simplified framework.

\subsection{Evaluation Map Images}\label{sec: eval map images}

We describe the algebraic structure of $B_\gamma(\alpha)$ when both $\gamma$ (the coefficient product; see Definition \ref{def: gamma bucket}) and $\alpha$ (the evaluation parameter; see Definition \ref{def: gamma bucket eval}) are restricted to $\alpha, \gamma \in \Omega_q$. Consider the case where $\gamma = 1$, so that
\begin{equation}
    B_1 = \left\lbrace \frac{1}{m}x + m : m \in \mathbb{F}_q^\ast \right\rbrace \subseteq \mathbb{F}_q[x].
\end{equation}
If we evaluate all lines $f \in B_1$ at $\alpha=1$, then
\begin{equation}
    B_1(1) = \left\lbrace \frac{1}{m} + m : m \in \mathbb{F}_q^\ast \right\rbrace \subseteq \mathbb{F}_q.
\end{equation}
It turns out that, up to scaling, $B_1(1)$ describes all $B_\gamma(\alpha)$ for $\gamma, \alpha \in \Omega_q$. To express the exact scaling factor, we first define a notion of square roots over $\Omega_q$. Note this is possible since $\Omega_q \subseteq \mathrm{QR}_q$ for all prime powers $q$, with exceptions\footnote{These exceptions do not affect the main result of Theorem \ref{thm: full-case bound}, since the bound is trivial over $\mathbb{F}_2, \mathbb{F}_4$. Indeed, $2 \log_2(q-2) - 4 \leq 0$ when $q = 2, 4$.} only when $q=2, 4$.

\begin{definition}[Square Roots in $\Omega_q$]\label{def: square roots in general}
Given $\alpha \in \Omega_q$, we define 
\begin{equation}
    r_\alpha \defeq \lbrace \beta \in \mathbb{F}_q^\ast : \beta^2 = \alpha \rbrace \subseteq \mathbb{F}_q^\ast.
\end{equation}
Since $r_\alpha = r_\gamma$ if and only if $\alpha = \gamma$, we may fix\footnote{The results of Section \ref{sec: restriction implications}, along with the accompanying discussion, hold for any arbitrary choice of square root representative for each $\alpha \in \Omega_q$. In Section \ref{sec: bounding subcase}, we employ a specific choice of such representatives which eases a counting argument in the proof of Theorem \ref{thm: output list size 2}, a key ingredient for the full result of Theorem \ref{thm: full-case bound}. The explicit construction of such a choice of square roots, along with their combinatorial properties, is given in Section \ref{sec: key bound}.} a canonical representative of $r_\alpha$ for each $\alpha \in \Omega_q$; we denote such a choice of representative $\sqrt{\alpha} \in r_\alpha$.
\end{definition}

Our goal will eventually be to understand  $B_\gamma(\alpha)$ in terms of $B_1(1)$ (see Lemma~\ref{lemma: bucket rescaling}).  We first need the following lemma.

\begin{lemma}\label{lemma: distribution of square roots}
Suppose that $q=p^e$. For all $\alpha \in \Omega_q$, 
    \begin{equation}
        \left\lbrace \alpha m^{-1} + m \; : \; m \in \mathbb{F}_q^\ast \right\rbrace = \left\lbrace \sqrt{\alpha} m^{-1} + \sqrt{\alpha} m \; : \; m \in \mathbb{F}_q^\ast \right\rbrace.
    \end{equation}
\end{lemma}
\begin{proof}
    Let $L, R$ denote the left, right hand sides of the equality, respectively. Observe that $\rho: m \mapsto m \sqrt{\alpha}$ is a permutation of $\mathbb{F}_q^\ast$. Hence
    \begin{align*}
        L &= \left\lbrace \alpha \rho(m)^{-1} + \rho(m) \; : \; m \in \mathbb{F}_q^\ast \right\rbrace = \left\lbrace \frac{\alpha}{m \sqrt{\alpha}} + m\sqrt{\alpha} \; : \; m \in \mathbb{F}_q^\ast \right\rbrace = R
    \end{align*}
    as desired.
\end{proof}

\begin{lemma}\label{lemma: bucket rescaling}
    For all $\alpha, \beta, \gamma, \delta \in \Omega_q$, 
    \begin{equation}
        B_\gamma(\alpha) = \frac{\sqrt{\gamma \alpha}}{\sqrt{\delta \beta}} \cdot B_\delta(\beta)
    \end{equation}
\end{lemma}
\begin{proof}
    Since $\beta,\delta \in \Omega_q$, we have $\delta \beta \in \Omega_q$. From Definition \ref{def: gamma bucket eval}, we have
    \begin{align}
        \frac{\sqrt{\gamma \alpha}}{\sqrt{\delta \beta}} \cdot B_\delta(\beta) &= \frac{\sqrt{\gamma \alpha}}{\sqrt{\delta \beta}} \cdot \left\lbrace \delta \beta m^{-1} + m \; : \; m \in \mathbb{F}_q^\ast \right\rbrace\\
        &= \frac{\sqrt{\gamma \alpha}}{\sqrt{\delta \beta}} \cdot \left\lbrace \sqrt{\delta \beta} m^{-1} + \sqrt{\delta \beta} m \; : \; m \in \mathbb{F}_q^\ast \right\rbrace\\
        &= \left\lbrace \sqrt{\gamma \alpha} m^{-1} + \sqrt{\gamma \alpha} m \; : \; m \in \mathbb{F}_q^\ast \right\rbrace\\
        &= \left\lbrace \gamma \alpha m^{-1} +  m \; : \; m \in \mathbb{F}_q^\ast \right\rbrace\\
        &= B_{\gamma}(\alpha)
    \end{align}
    where the second and fourth equalities follow from applying Lemma \ref{lemma: distribution of square roots}.
\end{proof}

We have the following corollary, which is an immediate special case of Lemma \ref{lemma: bucket rescaling}.

\begin{corollary}\label{corollary: scalar evolution}
    For all $\alpha, \gamma \in \Omega_q$, $B_\gamma(\alpha) = \sqrt{\gamma} \sqrt{\alpha} \cdot  B_1(1)$.
\end{corollary}

Such a multiplicative relationship between two images $B_\gamma(\alpha)$ and $B_1(1)$ allows us to work entirely over scalar multiples of $B_1(1) \subseteq \mathbb{F}_q$. The following observation gives the size of $B_1(1)$.

\begin{observation}\label{obs: size of B11}
    Let $\mathbb{F}_q = \mathbb{F}_{p^e}$ where $p$ is odd; or $p = 2$ and $e \geq 3$. Then
    \begin{equation*}
        \left\lvert B_1(1) \right\rvert = \begin{cases}
            (q+1)/2 & p \text{ odd}\\
            q/2 & p = 2, e \geq 3.
        \end{cases}
    \end{equation*}
\end{observation}
\begin{proof}
    Observe that, given $m, n \in \mathbb{F}_q^\ast$, we have $m + m^{-1} = n + n^{-1}$ if and only if $mn = 1$.
    \begin{itemize}
        \item If field characteristic $p$ is odd, there are two cases where $m = n$ and $mn = 1$, namely $m = n = \pm 1$. If $m \neq n$ and $mn = 1$, then $m,n \in \mathbb{F}_q^\ast \setminus \lbrace \pm 1 \rbrace$ form pairs of distinct multiplicative inverses. Hence we have
        \begin{equation*}
            \left\lvert B_1(1) \right\rvert = 2 + \frac{(q-1)-2}{2} = \frac{q+1}{2}
        \end{equation*}

        \item If the field characteristic $p = 2$ with extension degree $e \geq 3$, then there is only one case where $m = n$ and $mn = 1$, namely $m = n = 1$. By a similar reasoning as the preceding case, we have
        \begin{equation*}
            \left\lvert B_1(1) \right\rvert = 1 + \frac{(q-1)-1}{2} = \frac{q}{2}
        \end{equation*}
    \end{itemize}
\end{proof}

\subsection{Scalar Representatives of Lines}

We've seen that any two sets of evaluations $B_\gamma(\alpha), B_\delta(\beta)$ are related by a multiplicative factor for any $\alpha,\beta,\gamma, \delta \in \Omega_q$. Intuitively: why bother with different sets of points if they're just multiples of each other? This intuition leads us to ``zoom in'' further; we establish in this section that the multiplicative relation between sets stems directly from the same multiplicative relation among lines. 

First, we introduce the notation
\begin{equation}
    h_m^\gamma \defeq \frac{\sqrt{\gamma}}{m} x + m\sqrt{\gamma}
\end{equation}
so that, for every coefficient product $\gamma \in \Omega_q$, we may rewrite
\begin{equation}
    B_\gamma = \left\lbrace h_m^\gamma : m \in \mathbb{F}_q^\ast \right\rbrace\subseteq \mathbb{F}_q[x].
\end{equation}
Note that every element of $B_\gamma$ is then fully specified\footnote{We note that this part of the argument extends easily to higher-degree monomials.  For example, suppose we want to compute the product of all $k$ message symbols.  Suppose $\gamma \in \mathbb{F}_q^\ast$ is a degree $k$ residue, and that $\mathbf{f}\in (\mathbb{F}_q^\ast)^{k}$ satisfies $\prod_{i=0}^{k-1} f_i = 1$. Then any polynomial $h = \sum_{i=0}^{k-1} h_i x^i$ with $\prod_{i=0}^{k-1} h_i = \gamma$ may be written as $\gamma^{1/k} f(x)$ for some $f$. However, extending later arguments require a characterization of arbitrary degree $k$ residue subgroups, which is left to future work.} by some $m \in \mathbb{F}_q^\ast$ and $\gamma \in \Omega_q$.

\begin{observation}\label{obs: everything is a multiple of 1-lines}
    For all $m \in \mathbb{F}_q^\ast$, $\gamma \in \Omega_q$ specifying $h_m^\gamma \in B_\gamma$, there exists a unique $g_m \in B_1$ such that $h_m^\gamma = \sqrt{\gamma}g_m$. Furthermore, $g_m$ is given by 
    \begin{equation}
        g_m \defeq h_m^1 = \frac{1}{m} x + m.
    \end{equation}
\end{observation}

The observation follows trivially from $h_m^\gamma = \sqrt{\gamma}\cdot h_m^1$, but its significance is that we need only consider $B_1 \subseteq \mathbb{F}_q[x]$, since any line $h_m^\gamma \in B_\gamma$ is simply a multiple of $g_m \in B_1$. 

\begin{observation}\label{obs: scale the eliminators}
    For all $m \in \mathbb{F}_q^\ast$, $\alpha,\gamma  \in \Omega_q$, and $T \subseteq \mathbb{F}_q$,
    \begin{equation}
        h_m^\gamma(\alpha) \in T \; \iff \; g_m(\alpha) \in \frac{1}{\sqrt{\gamma}}T.
    \end{equation}
\end{observation}

Hence, we can check whether any line $h_m^\gamma$, evaluated at $\alpha \in \mathbb{F}_q$, is in a set $T$ just by considering whether $g_m(\alpha) \in B_1(\alpha)$ is in a multiple of $T$. 

Corollary \ref{corollary: scalar evolution} shows that we may map the entire evaluation image of $B_\gamma$ at some point $x = \alpha$ to its image at another evaluation point $x=\beta$, simply by multiplying $B_\gamma(\alpha)$ by $\sqrt{\beta}/\sqrt{\alpha}$. But this macroscopic view does not elucidate what is happening to individual lines. We show that multiplying an evaluation map image $B_\gamma(\alpha)\subseteq \mathbb{F}_q$ by some scalar is equivalent to a permutation of $B_\gamma\subseteq \mathbb{F}_q[x]$. 

\begin{definition}[$\alpha$-Relabeling function]\label{def: relabeling fn}
    For all $\alpha,\gamma \in \Omega_q$, define the mapping $\phi_\alpha^\gamma : B_\gamma \to B_\gamma$ by $h_m^\gamma \mapsto h_{m\sqrt{\alpha}}^\gamma$. Explicitly,
    \begin{equation}
        \frac{\sqrt{\gamma}}{m} x + m\sqrt{\gamma} \mapsto \frac{\sqrt{\gamma}}{(m \sqrt{\alpha})} x + (m\sqrt{\alpha})\sqrt{\gamma}.
    \end{equation}
    Given $h_m^\gamma \in B_\gamma$, we call its image $\phi^\gamma_\alpha[h_m^\gamma] \in B_\gamma$ the $\alpha$-relabel of $h_m^\gamma$.
\end{definition}

Note that $\phi_\alpha^\gamma$ is mapping one line in $B_\gamma$ to another line also in $B_\gamma$. The following observation shows why we have chosen to call it a relabeling; the proof follows from the fact that for any $\alpha \in \Omega_q$, the map $m \mapsto m\sqrt{\alpha}$ is a permutation of $\mathbb{F}_q^\ast$, and that each element  $m \in \mathbb{F}_q^\ast$ corresponds to a distinct element $h_m^\gamma \in B_\gamma$.
\begin{observation}\label{obs: i relabel is permutation}
    For all $\alpha,\gamma \in \Omega_q$, $\phi_\alpha^\gamma : B_\gamma \to B_\gamma$ is a permutation of $B_\gamma$.
\end{observation}

Given some $h_m^\gamma \in B_\gamma$ and $\alpha \in \Omega_q$, we denote by $\phi_\alpha^\gamma[h_m^\gamma](\alpha)$ the evaluation of the $\alpha$-relabel of $h_m^\gamma$ at an evaluation point $\alpha$. We show that the multiplicative relationship between lines in $B_\gamma$ and lines in $B_1$ (Observation \ref{obs: everything is a multiple of 1-lines}) is preserved by this relabeling.

\begin{lemma}
    Let $\alpha,\gamma \in \Omega_q$ and $h_m^\gamma \in B_\gamma$. Then
    \begin{equation}
        \phi_\alpha^\gamma\left[ h_m^\gamma \right](\alpha) = \sqrt{\gamma} \cdot \phi_\alpha^1 \left[ g_m \right] (\alpha).
    \end{equation}
\end{lemma}

\begin{proof}
    On the left-hand side, observe that, applying Definition \ref{def: relabeling fn} and Observation \ref{obs: everything is a multiple of 1-lines},
    \begin{align}
        \phi_\alpha^\gamma\left[ h_m^\gamma \right](\alpha) &= h_{m\sqrt{\alpha}}^\gamma(\alpha) = \sqrt{\gamma}\cdot h_{m\sqrt{\alpha}}^1(\alpha) = \sqrt{\gamma}\cdot g_{m\sqrt{\alpha}}(\alpha).
    \end{align}
    On the other hand, from the right we have
    \begin{align}
        \sqrt{\gamma}\cdot \phi_\alpha^1 \left[ g_m\right] (\alpha) = \sqrt{\gamma} \cdot g_{m \sqrt{\alpha}} (\alpha)
    \end{align}
    as desired.
\end{proof}

\begin{lemma}
    Let $\alpha \in \Omega_q$ and $g_m \in B_1$. Then $\phi_\alpha^1\left[ g_m \right](\alpha) = \sqrt{\alpha} g_m(1)$.
\end{lemma}
\begin{proof}
    Recalling the definition of $g_m \in B_1$ within Observation \ref{obs: everything is a multiple of 1-lines}, observe
    \begin{align}
        \phi_\alpha^1\left[g_m\right](\alpha) &= g_{m\sqrt{\alpha}}(\alpha)\\
        &= \left.\frac{1}{m\sqrt{\alpha}}x + m \sqrt{\alpha}\right\rvert_{x=\alpha}\\
        &= \sqrt{\alpha}\left(\frac{1}{m} + m\right)\\
        &= \sqrt{\alpha} \cdot g_m(1).
    \end{align}
\end{proof}

Combining these two results gives the following theorem.

\begin{theorem}\label{thm: relabel and flatten}
    Let $\alpha, \gamma \in \Omega_q$ and $h_m^\gamma \in B_\gamma$. Then
    \begin{equation}
        \phi_\alpha^\gamma\left[ h_m^\gamma \right](\alpha) = \sqrt{\gamma \alpha} \cdot g_m(1).
    \end{equation}
\end{theorem}

We extend Theorem \ref{thm: relabel and flatten} with the following corollary, whose proof is immediate. This is the $\alpha$-relabel analogue of Observation \ref{obs: scale the eliminators}, and it is the tool we will need to reduce a subset of QM to a simpler algorithm in the following section.

\begin{corollary}\label{cor: label land eliminators}
    For all $m \in \mathbb{F}_q^\ast$, $\alpha, \gamma \in \Omega_q$, and $T \subseteq \mathbb{F}_q$,
    \begin{equation}
        \phi_\alpha^\gamma \left[ h_m^\gamma \right] (\alpha) \in T \quad\iff\quad g_m(1) = m + m^{-1} \in \frac{1}{\sqrt{\gamma \alpha}} T.
    \end{equation}
\end{corollary}

Together, Theorem \ref{thm: relabel and flatten} and Corollary \ref{cor: label land eliminators} capture the heart of our simplification; we need not worry ourselves with such objects such as $B_\gamma \subseteq \mathbb{F}_q[x]$ or $h_m^\gamma \in B_\gamma$. By using $\alpha$-relabels, during a round which queries a server indexed by $\alpha \in \Omega_q$:

\begin{enumerate}
    \item We can send a line $h_m^\gamma$ to its $\alpha$-relabel $h_{m\sqrt{\alpha}}^\gamma$; then

    \item We evaluate $h_{m\sqrt{\alpha}}^\gamma(\alpha)$ as a multiple of $g_m(1)$.
\end{enumerate}

We need only track $\alpha, \gamma \in \Omega_q$ and $m \in \mathbb{F}_q^\ast$. \textit{Everything} is some explicit scalar multiple of the known value $g_m(1) \in B_1(1) \subseteq \mathbb{F}_q$, so we can work simply over $B_1(1) \subseteq \mathbb{F}_q$.

\begin{remark}
    Note that the $\alpha$-relabeling of two distinct lines $h_m^\gamma \in B_\gamma$ and $h_m^\delta \in B_\delta$ preserves their scalar multiplicative relationship at the evaluation point $\alpha \in \Omega_q$. Explicitly, observe that the lines $h_m^\gamma, h_m^\delta$ are scalar multiples of each other:
    \begin{align}
        \frac{\sqrt{\delta}}{\sqrt{\gamma}} \cdot h_m^\gamma(\alpha) &= \frac{\sqrt{\delta }}{\sqrt{\gamma }}\left(\frac{\sqrt{\gamma}}{m} \alpha + m \sqrt{\gamma}\right)= \frac{\sqrt{\delta}}{m} \alpha+ m\sqrt{\delta} = h_m^\delta (\alpha).
    \end{align}
    The exact relation holds for their $\alpha$-relabels, evaluated at $\alpha$; explicitly, applying Theorem \ref{thm: relabel and flatten}, we see
    \begin{align*}
        \frac{\sqrt{\delta}}{\sqrt{\gamma}} \cdot \phi_\alpha^{\gamma} \left[ h_m^\gamma \right] (\alpha)= \frac{\sqrt{\delta}}{\sqrt{\gamma}}\left( \sqrt{\gamma \alpha}\cdot g_m(1)\right) = \sqrt{\delta \alpha} \cdot g_m(1) = \phi_\alpha^{\delta} \left[ h_m^\delta \right] (\alpha).
    \end{align*}
\end{remark}

\subsection{Reinterpreting Leakage Functions}

We now map the sets $T_i$ (defining bit-valued leakage functions; see Algorithm \ref{alg: QM Full}) to our new, relabeled setting. Suppose a round $i \in [t]$ is associated with a server indexed by $\alpha_i \in \Omega_q$ (i.e., the server holding the evaluation $f(\alpha_i)$). A candidate $h_m^\gamma \in B_\gamma$ is deemed inconsistent with the $i$th leakage bit $b$ (hence, eliminated as a candidate for $f$) if $b=0$ and $h_m^\gamma(\alpha_i) \in T_i$, or if $b=1$ and $h_m^\gamma(\alpha_i) \in \mathbb{F}_q \setminus T_i$.  

\begin{definition}\label{def: converted eliminators}
Given some round $i \in [t]$ associated with $\alpha_i \in \Omega_q$ and $T_i\subseteq \mathbb{F}_q$, define
\begin{equation}
    U_i \defeq \bigcup_{\gamma \in \Omega_q} \left\lbrace \phi_{\alpha_i}^\gamma \left[ h_m^\gamma \right] \; : \; h_m^\gamma \in  B_\gamma, \; h_m^\gamma(\alpha_i)\in T_i \right\rbrace \subseteq \mathbb{F}_q[x]. 
\end{equation}
(That is, $U_i$ is the set of $\alpha_i$-relabels $\phi_{\alpha_i}^\gamma\left[ h_m^\gamma\right]$ of lines $h_m^\gamma$ satisfying $h_m^\gamma(\alpha_i) \in T_i$.) Furthermore, define
\begin{equation}\label{eqn: transformed eliminators}
    V_i \defeq \frac{1}{\sqrt{\alpha_i}}\left\lbrace f(\alpha_i) : f \in U_i \right\rbrace.
\end{equation}
\end{definition}
We now show that these $V_i$, as constructed in \eqref{eqn: transformed eliminators}, implement a similar action as $T_i$, but from the viewpoint of $\alpha_i$-relabels.

\begin{theorem}\label{thm: eliminator conversion}
    As per Definition~\ref{def: restricted parameter regime}, consider an instance of QM restricted to coefficient products and Reed-Solomon evaluation points that both lie in $\Omega_q$. During a QM round indexed by $i\in [t]$, let $m \in \mathbb{F}_q^\ast$, $\gamma, \alpha_i \in \Omega_q$; then
    \begin{equation}
        h_m^\gamma(\alpha_i) \in T_i \quad \implies \quad g_m(1) \in \frac{1}{\sqrt{\gamma}}V_i.
    \end{equation}
\end{theorem}

In words, if a line $h_m^\gamma$ is removed from consideration in round $i$ of a QM scheme, its corresponding point $g_m(1)$ must be contained in a rescaling of $V_i$. Note the similarity of this statement with Corollary \ref{cor: label land eliminators}. The significance is clear when combined with the previous section: 
\begin{enumerate}
    \item By Observation \ref{obs: i relabel is permutation}, $\phi_{\alpha_i}^\gamma$ bijects each $h_m^\gamma \in B_\gamma$ to an $\alpha_i$-relabel $\phi_{\alpha_i}^\gamma\left[ h_m^\gamma \right]\in B_\gamma$.
    \item By Theorem \ref{thm: relabel and flatten}, we can easily compute the value of $\alpha_i$-relabels at evaluation parameter $\alpha_i$ as a rescaling of $g_m(1)\in B_1(1)$.
\end{enumerate}
Thus, for each $\gamma \in \Omega_q$, there is an explicit mapping of each $h_m^\gamma \in B_\gamma$ to $g_m \in B_1$ (see Observation \ref{obs: everything is a multiple of 1-lines}), and we use this correspondence to map $B_\gamma \subseteq \mathbb{F}_q[x]$ onto (a rescaling of) $B_1(1)\subseteq \mathbb{F}_q$. This suggests that we make a copy of $B_1(1)$ for each $\gamma \in \Omega_q$, each representing some $B_\gamma$, and mirror the removal of lines from $B_\gamma$ by the removal of points from $B_1(1)$. Theorem \ref{thm: eliminator conversion} precisely tells us that this is possible. This high-level description captures the main idea of an algorithm we make explicit in the next section.

\begin{proof}[Proof of Theorem \ref{thm: eliminator conversion}]
    Suppose that we query a server indexed by $\alpha_i \in \Omega_q$ holding $f(\alpha_i)$, whose $i$th leakage bit $b_i$ is inconsistent with some candidate $h_m^\gamma \in B_\gamma$ for the underlying Reed-Solomon message polynomial $f(x)$. Assume without loss of generality that $b_i = 0$. Then removing $h_m^\gamma$ from consideration occurs precisely when $h_m^\gamma(\alpha_i) \in T_i$. By the construction of $U_i, V_i$,
    \begin{align}\label{eqn: eliminator conversion iffs}
        h_m^\gamma(\alpha_i) \in T_i &\implies \phi_{\alpha_i}^\gamma \left[ h_m^\gamma \right] \in U_i \implies \frac{1}{\sqrt{\alpha_i}} \cdot \phi_{\alpha_i}^\gamma\left[ h_m^\gamma \right](\alpha_i) \in V_i.
    \end{align}
    By Theorem \ref{thm: relabel and flatten}, we rewrite
    \begin{equation}
        \frac{1}{\sqrt{\alpha_i}} \cdot \phi_{\alpha_i}^\gamma\left[ h_m^\gamma \right](\alpha_i) = \frac{1}{\sqrt{\alpha_i}}\left(\sqrt{\gamma \alpha_i} \cdot g_m(1) \right) = \sqrt{\gamma}\cdot g_m(1)
    \end{equation}
    which, substituted into \eqref{eqn: eliminator conversion iffs}, yields the desired result.
\end{proof}

\section{A Key Bound}\label{sec: key bound}
The primary goal of this section is to prove the following theorem, which underpins our later analysis.

\begin{theorem}\label{thm: size of scaled pairs set}
    Let $\mathbb{F}_q = \mathbb{F}_{p^e}$ where $p$ is an odd prime; or $p = 2$ and $e \geq 3$. Then there exists $a \in \Omega_q$ and $b \in \mathbb{F}_q^\ast \setminus \Omega_q$ such that $\lbrace a, b \rbrace \subseteq B_1(1)$ (Definition \ref{def: gamma bucket eval}) along with a choice of square roots (Definition \ref{def: square roots in general}) so that, for all $\delta \in \Omega_q$,
    \begin{equation*}
         \left\lvert\bigcup_{\gamma \in \Omega_q \setminus \lbrace \delta \rbrace} \sqrt{\gamma}\cdot \lbrace a, b \rbrace\right\rvert = \begin{cases}
             q-3 & p > 2\\
             q-4 & p = 2, e \geq 3
         \end{cases}
    \end{equation*}
\end{theorem}

Theorem \ref{thm: size of scaled pairs set} is proven below in three parts. Observation \ref{obs: scale pairs when -1 not in QRq} addresses the case where $p = 3 \text{ mod } 4$ and $e$ odd, and Observation \ref{obs: scale pairs when -1 in QRq} addresses the case where $p = 1 \text{ mod } 4$ or $e$ even. Finally, Observation \ref{obs: scale pairs over binary extension} addresses the case where $p=2$ and $e \geq 3$.

\subsection{\texorpdfstring{Odd prime characteristic $p$}{Odd prime characteristic p}}

In this regime, $\Omega_q = \mathrm{QR}_q$. A critical property of $B_1(1)$ is that the set always contains at least one element of $\mathrm{QR}_q$ and another from $\mathbb{F}_q^\ast \setminus \mathrm{QR}_q$.

\begin{lemma}\label{lem: B11 has one QR and one non QR}
    Let $\mathbb{F}_q = \mathbb{F}_{p^e}$ with odd prime characteristic $p$ and field order\footnote{The condition $q \neq 5$ is necessary as Lemma \ref{lem: B11 has one QR and one non QR} is simply untrue in that setting. Indeed, over $\mathbb{F}_5$, $B_1(1) = \lbrace 0, 2, 3 \rbrace$ and $QR_5 = \lbrace 1,4 \rbrace$ have no elements in common.} $q > 5$. Define
    \begin{equation}
        B_1(1) \defeq \left\lbrace m^{-1} + m : m \in \mathbb{F}_q^\ast \right\rbrace \subseteq \mathbb{F}_q.
    \end{equation}
    Then there exist distinct, non-zero $a, b \in B_1(1) \setminus \lbrace 0 \rbrace$ such that $a \in \mathrm{QR}_q$ and $b \not\in \mathrm{QR}_q$.
\end{lemma}

\begin{proof}
    Suppose towards a contradiction that $B_1(1)$ does not contain a pair of distinct elements $a,b \in B_1(1) \setminus \lbrace 0 \rbrace$ such that $a \in \mathrm{QR}_q$ and $b \not\in \mathrm{QR}_q$. This implies that all elements of $B_1(1)\setminus \lbrace 0 \rbrace$ are either quadratic residues or quadratic non-residues. Recall the proof of Lemma \ref{lemma: distribution of square roots} established that $|B_1(1)| = (q+1)/2$; with the preceding observations, this implies
    \begin{align}
        \left\lvert \sum_{b \in B_1(1)} \chi(b) \right\rvert &= \left\lvert B_1(1) \setminus \lbrace 0 \rbrace \right\rvert \\
        &\geq \frac{q+1}{2}-1 \\
        & = \frac{q-1}{2}. \label{eqn: lower bound on partial char sum}
    \end{align}

    On the other hand, let $g(x) \defeq x^{-1} + x$; then for all $m,m' \in \mathbb{F}_q^\ast$, $g(m) = b$ if and only if $g(m^{-1}) = b$, and if $m \neq m'$ then  $g(m) = g(m')$ if and only if $m' = m^{-1}$. Hence, by the definition of $B_1(1)$ and the assumption that $B_1(1) \setminus \lbrace 0 \rbrace$ contains only either all quadratic residues or quadratic non-residues,
    \begin{align}
        \left\lvert \sum_{b \in B_1(1)} \chi(b) \right\rvert &= \frac{1}{2}\left\lvert \sum_{m \in \mathbb{F}_q^\ast} \chi(m^{-1} + m) \right\rvert \\
        &= \frac{1}{2}\left\lvert \sum_{m \in \mathbb{F}_q^\ast} \chi(m^{-1}) \chi(1+m^2) \right\rvert\\
        &= \frac{1}{2}\left\lvert \sum_{m \in \mathbb{F}_q} \chi(m(m^2 + 1)) \right\rvert \label{eqn: B11 complete char sum}
    \end{align}
    where the second, third equalities follow from Observation \ref{obs: two quadratic char props}. The summation in \eqref{eqn: B11 complete char sum} is over the entire field $\mathbb{F}_q$ rather than only $\mathbb{F}_q^\ast$, making it a complete character sum (Definition \ref{def: complete quad char sum}). Furthermore, $f(m) \defeq m(m^2 + 1)$ is square-free\footnote{Indeed, it suffices to observe that its factor $g(m) \defeq m^2 + 1$ is square free, as its formal derivative $g'(m) = 2m$ has only $m=0$ as a root, while $g(0) = 1$.} when $p > 2$. Hence, we may apply Theorem \ref{thm: weil char sum estimate} to \eqref{eqn: B11 complete char sum} and observe
    \begin{align}
        \left\lvert \sum_{b \in B_1(1)} \chi(b) \right\rvert = \frac{1}{2}\left\lvert \sum_{m \in \mathbb{F}_q} \chi(m(m^2 + 1)) \right\rvert \leq \frac{(\deg(f)-1)\sqrt{q}}{2} = \sqrt{q}.
    \end{align}
    Combined with \eqref{eqn: lower bound on partial char sum}, this implies
    \begin{equation}
        \frac{q-1}{2} \leq  \sqrt{q} \quad \iff \quad q-1 \leq 2 \sqrt{q}
    \end{equation}
    which is false for any prime power $q \geq 7$. It remains only to consider $q=3$, in which case $B_1(1) = \lbrace 1, 2 \rbrace$ and $\mathrm{QR}_3 = \lbrace 1 \rbrace$ clearly satisfy Lemma \ref{lem: B11 has one QR and one non QR}, since $B_1(1)$ contains precisely one quadratic residue and one non-residue.
\end{proof}

\subsubsection{\texorpdfstring{$p = 3 \text{ mod } 4$, $e$ odd}{p = 3 mod 4, e odd}}

In the case $p = 3 \text{ mod } 4$ and $e$ is odd, we choose square root representatives over $\Omega_q = \mathrm{QR}_q$ as follows.

\begin{claim}[Square roots over $\mathrm{QR}_q$]\label{claim: nice square roots over QRq}
    Fix prime $p>2$ and odd extension degree $e$ with $p-3 \mod 4 = 0$; set $q=p^e$.
For all $\alpha \in \mathrm{QR}_q$, there exists a unique $\beta \in \mathrm{QR}_q$ such that $\beta^2 = \alpha$; we define $\sqrt{\alpha} \defeq \beta$.
\end{claim}

To show Claim \ref{claim: nice square roots over QRq}, we will need the following lemma.

\begin{lemma}\label{lem: when is -1 a residue}
    Let $q = p^e$ for prime $p$. Then $p = 3 \text{ mod } 4$ and $e\in \mathbb{Z}^+$ odd if and only if $-1 \not\in \mathrm{QR}_q$.
\end{lemma}

\begin{proof}
    The quadratic character (Definition \ref{def: quadratic char}) satisfies $\chi(x) = x^{(q-1)/2}$. Furthermore, $q = p^e = 3 \text{ mod } 4$ if and only if $(q-1)/2$ is odd; in turn, $(q-1)/2$ is odd if and only if $\chi(-1) = (-1)^{(q-1)/2} = -1$. The latter is equivalent to $-1 \not\in \mathrm{QR}_q$.
\end{proof}

\begin{proof}[Proof of Claim \ref{claim: nice square roots over QRq}]
    Suppose $\alpha \in \mathrm{QR}_q$, and that $\beta, -\beta \in \mathbb{F}_q^\ast$ both square to $\alpha$. Since $-1 \not\in \mathrm{QR}_q$, we have $\chi(-\beta) = \chi(-1)\chi(\beta) = - \chi(\beta)$, which implies that precisely one of $\beta, -\beta$ lies in $\mathrm{QR}_q$.
\end{proof}

We then have the following observation.

\begin{observation}\label{obs: scale pairs when -1 not in QRq}
    Let $q= p^e$ for odd extension degree $e$ and prime characteristic $p$ satisfying $p = 3 \mod 4$. For all $\gamma \in \mathrm{QR}_q$, define $\sqrt{\gamma}$ as in Claim \ref{claim: nice square roots over QRq}. Let $a \in \mathrm{QR}_q$ and $b \in \mathbb{F}_q^\ast \setminus \mathrm{QR}_q$ satisfying $\lbrace a, b \rbrace \subseteq B_1(1)$; such a choice exists by Lemma \ref{lem: B11 has one QR and one non QR}. Then for all $\delta \in \mathrm{QR}_q$,
    \begin{equation}
    \left\lvert \bigcup_{\gamma \in \mathrm{QR}_q \setminus \lbrace \delta \rbrace} \sqrt{\gamma}\cdot \lbrace a, b \rbrace \right\rvert = q-3.
    \end{equation}
\end{observation}

\begin{proof}
    For each $\gamma \in \mathrm{QR}_q$, we have $\sqrt{\gamma} \in \mathrm{QR}_q$ so that $\sqrt{\gamma} \cdot a \in \mathrm{QR}_q$ and $\sqrt{\gamma} \cdot b \in \mathbb{F}_q^\ast \setminus \mathrm{QR}_q$. Hence,
    \begin{equation}
        \bigcup_{\gamma \in \mathrm{QR}_q \setminus \lbrace \delta \rbrace} \sqrt{\gamma}\cdot \lbrace a, b \rbrace = \left\lbrace \sqrt{\gamma} \cdot a : \gamma \in \mathrm{QR}_q \setminus \lbrace \delta \rbrace \right\rbrace \cup \left\lbrace \sqrt{\gamma} \cdot b : \gamma \in \mathrm{QR}_q\setminus \lbrace \delta \rbrace \right\rbrace
    \end{equation}
    is the union of two disjoint sets. It follows that
    \begin{align}
        \left\lvert  \bigcup_{\gamma \in \mathrm{QR}_q \setminus \lbrace \delta \rbrace} \sqrt{\gamma}\cdot \lbrace a, b \rbrace \right\rvert = \left(\frac{q-1}{2}-1 \right) \cdot 2 = q-3.
    \end{align}
\end{proof}

\subsubsection{\texorpdfstring{$p=1 \text{ mod } 4$, or $e$ even}{p = 1 mod 4, or e even}}

In this regime, we have $\Omega_q = \mathrm{QR}_q$, and Lemma \ref{lem: when is -1 a residue} shows $-1 \in \mathrm{QR}_q$. The closure of $\mathrm{QR}_q$ under multiplication implies the elements of $\mathrm{QR}_q$ may be decomposed into pairs of additive inverses; by symmetry, the same holds for $\mathbb{F}_q^\ast \setminus \mathrm{QR}_q$. Hence, for all $\gamma \in \mathrm{QR}_q$, precisely one of the following holds: either
\begin{enumerate}
    \item for all $g \in \mathbb{F}_q^\ast$ satisfying $g^2 = \gamma$, both $\pm g \in \mathrm{QR}_q$; or

    \item for all $g \in \mathbb{F}_q^\ast$ satisfying $g^2 = \gamma$, both $\pm g \in \mathbb{F}_q^\ast \setminus \mathrm{QR}_q$.
\end{enumerate}

Suppose $\gamma \in \mathrm{QR}_q$ satisfies Case \#1 above; that is, both possible representatives of $\sqrt{\gamma}$ are themselves quadratic residues. In this case, $\gamma$ must be a quartic residue. 

\begin{definition}\label{def: quartic residues}
    We define the (multiplicative group of) quartic residues over $\mathbb{F}_q$ as the set $R_q(4) \defeq\lbrace x^4 : x \in \mathbb{F}_q^\ast \rbrace$. Note that
    \begin{equation*}
        R_q(4) = \left\lbrace x^2 : x \in \mathrm{QR}_q \right\rbrace.
    \end{equation*}
\end{definition}

\begin{observation}\label{obs: flip flop set}
    Let $a \in \mathrm{QR}_q$ and $b \in \mathbb{F}_q^\ast \setminus \mathrm{QR}_q$. For each $\gamma \in R_q(4)$, fix some choice of $g \in \mathrm{QR}_q$ such that $g^2 = \gamma$; denote $\sqrt{\gamma} \defeq g$. Then
    \begin{equation*}
        \Upsilon \defeq \left\lbrace \frac{a}{b} \sqrt{\gamma} : \gamma \in R_q(4) \right\rbrace = \left\lbrace \frac{b}{a} \sqrt{\gamma} : \gamma \in R_q(4) \right\rbrace \subseteq \mathbb{F}_q^\ast \setminus \mathrm{QR}_q;
    \end{equation*}
    furthermore, $\Upsilon$ does not contain any pairs of additive inverses. 
\end{observation}

\begin{proof}
    Note that $(b/a)^4 \in R_q(4)$. Since $R_q(4)$ is a subgroup of $\mathrm{QR}_q$, which in turn is a subgroup of the cyclic group $\mathbb{F}_q^\ast$, it follows that the map $x \mapsto (b/a)^4 x$ is a permutation of $R_q(4)$. Hence,
    \begin{align}
        \left\lbrace \frac{a}{b} \sqrt{\gamma} : \gamma \in R_q(4) \right\rbrace &= \left\lbrace \frac{a}{b} \sqrt{\frac{b^4}{a^4}\gamma} : \gamma \in R_q(4) \right\rbrace\\
        &= \left\lbrace \frac{b}{a} \sqrt{\gamma} : \gamma \in R_q(4) \right\rbrace.
    \end{align}
    Next, since $a/b \in \mathbb{F}_q^\ast \setminus \mathrm{QR}_q$ and $\sqrt{\gamma} \in \mathrm{QR}_q$, their product is not in $\mathrm{QR}_q$:
    \begin{equation}
        \Upsilon = \left\lbrace \frac{a}{b} \sqrt{\gamma} : \gamma \in R_q(4) \right\rbrace  \subseteq \mathbb{F}_q^\ast \setminus \mathrm{QR}_q.
    \end{equation}
    Finally, $\Upsilon$ cannot contain any pair of additive inverses. Indeed, for $\gamma, \delta \in R_q(4)$,
    \begin{align}
         \frac{a}{b} \sqrt{\gamma} = -\frac{a}{b} \sqrt{\delta} \quad \iff \quad \sqrt{\gamma} = -\sqrt{\delta}.
    \end{align}
    Squaring both sides yields $\gamma = \delta$. It remains only to observe that each element of $(a/b)\sqrt{\gamma} \in \Upsilon$ corresponds to exactly one element $\gamma \in R_q(4)$.
\end{proof}

For all $\gamma \in R_q(4) \subseteq \mathrm{QR}_q$, we may then select $\sqrt{\gamma}$ to be in $\mathrm{QR}_q$; we now show how to select $\sqrt{\omega}$ for $\omega \in \mathrm{QR}_q \setminus R_q(4)$.

\begin{observation}\label{obs: nice non-quartic square root reps}
    For all $\omega \in \mathrm{QR}_q \setminus R_q(4)$, there exists some element
    \begin{equation*}
        g \in \mathbb{F}_q^\ast \setminus \left( \mathrm{QR}_q \cup \Upsilon\right)
    \end{equation*}
    such that $g^2 = \omega$; hence, we may denote $\sqrt{\omega}\defeq g$ for such a choice. 
\end{observation}

\begin{proof}
    For all $\omega \in \mathrm{QR}_q \setminus R_q(4)$, all $g \in \mathbb{F}_q^\ast$ satisfying $g^2 = \omega$ must satisfy $\pm g \in \mathbb{F}_q^\ast \setminus \mathrm{QR}_q$. It suffices to show that at least one of $\pm g$ is not contained in $\Upsilon \subseteq \mathbb{F}_q^\ast \setminus \mathrm{QR}_q$. However, this follows from Observation \ref{obs: flip flop set}, as $\Upsilon$ does not contain any pair of additive inverses.
\end{proof}

We now summarize the choice of square root representatives when $p = 1 \mod 4$ or $e$ is even.

\begin{definition}\label{def: extended square root representatives}
    Let $q = p^e$ for even extension degree $e$; or odd extension degree $e$ when prime characteristic $p$ satisfies $p = 1 \mod 4$. Then we define the square root operation $\sqrt{\cdot} : \mathrm{QR}_q \to \mathbb{F}_q$ by the following:
    \begin{enumerate}
        \item for all $\gamma \in R_q(4)$, fix arbitrarily a choice $\sqrt{\gamma} \defeq g \in \mathrm{QR}_q$ satisfying $g^2 = \gamma$; and

        \item for all $\omega \in \mathrm{QR}_q \setminus R_q(4)$, fix arbitrarily a choice $\sqrt{\omega}\defeq g \in \mathbb{F}_q^\ast \setminus \left(\mathrm{QR}_q \cup \Upsilon \right)$ satisfying $g^2 = \omega$; such a choice exists by Observation \ref{obs: nice non-quartic square root reps}.
    \end{enumerate}
\end{definition}

\begin{observation}\label{obs: scale pairs when -1 in QRq}
    Let $q = p^e$ for even extension degree $e$; or odd extension degree $e$ when prime characteristic $p$ satisfies $p = 1 \mod 4$. For $\gamma \in \mathrm{QR}_q$, define $\sqrt{\gamma}$ as in Definition \ref{def: extended square root representatives}. Let $a \in \mathrm{QR}_q$ and $b \in \mathbb{F}_q^\ast \setminus \mathrm{QR}_q$ satisfying $\lbrace a,b \rbrace \subseteq B_1(1)$; such a choice exists by Lemma \ref{lem: B11 has one QR and one non QR}. Then, for all $\delta \in \mathrm{QR}_q$,
    \begin{equation}
    \left\lvert \bigcup_{\gamma \in \mathrm{QR}_q \setminus \lbrace \delta \rbrace} \sqrt{\gamma}\cdot \lbrace a, b \rbrace \right\rvert = q-3.
\end{equation}
\end{observation}

\begin{proof}
Observe that
\begin{align}
    \bigcup_{\gamma \in \mathrm{QR}_q \setminus \lbrace \delta \rbrace} \sqrt{\gamma}\cdot \lbrace a, b \rbrace &= \left( \bigcup_{\gamma \in R_q(4)\setminus \lbrace \delta \rbrace} \sqrt{\gamma} \cdot \lbrace a, b \rbrace \right) \cup \left( \bigcup_{\omega \in \mathrm{QR}_q \setminus \left(R_q(4) \cup \lbrace \delta \rbrace \right)} \sqrt{\omega}\cdot \lbrace a, b \rbrace \right) 
\end{align}
may be decomposed as the union of the two following sets:
\begin{equation}\label{eqn: group up the residues}
    \left\lbrace \sqrt{\gamma} \cdot a : \gamma \in R_q(4) \setminus \lbrace \delta \rbrace \right\rbrace \cup \left\lbrace \sqrt{\omega} \cdot b : \omega \in \mathrm{QR}_q \setminus \left(R_q(4) \cup \lbrace \delta \rbrace \right) \right\rbrace \subseteq \mathrm{QR}_q
\end{equation}
and
\begin{equation}\label{eqn: group up the non-residues}
    \left\lbrace \sqrt{\gamma} \cdot b : \gamma \in R_q(4) \setminus \lbrace \delta \rbrace \right\rbrace \cup \left\lbrace \sqrt{\omega} \cdot a : \omega \in \mathrm{QR}_q \setminus \left(R_q(4) \cup \lbrace \delta \rbrace \right) \right\rbrace \subseteq \mathbb{F}_q^\ast \setminus \mathrm{QR}_q.
\end{equation}
Let us first examine \eqref{eqn: group up the residues}; we claim that the two sides of the union are disjoint. Indeed, 
\begin{equation}
    \sqrt{\gamma} \cdot a = \sqrt{\omega} \cdot b \quad \iff \quad \sqrt{\omega} = \frac{a}{b} \sqrt{\gamma}
\end{equation}
which by Definition \ref{def: extended square root representatives} is not possible. Similarly, the two sides of the union in \eqref{eqn: group up the non-residues} are disjoint, since\begin{equation}
    \sqrt{\gamma} \cdot b = \sqrt{\omega} \cdot a \quad \iff \quad \sqrt{\omega} = \frac{b}{a} \sqrt{\gamma}
\end{equation}
and, again, by Definition \ref{def: extended square root representatives}, does not occur. Hence, since either $\delta \in R_q(4)$ or $\delta \in \mathrm{QR}_q \setminus R_q(4)$,
\begin{equation}
    \left\lvert \bigcup_{\gamma \in \mathrm{QR}_q \setminus \lbrace \delta \rbrace} \sqrt{\gamma}\cdot \lbrace a, b \rbrace \right\rvert = 2 \cdot \left(\frac{q-1}{2} - 1 \right) = q-3.
\end{equation}
\end{proof}

\subsection{\texorpdfstring{Binary characteristic $p=2$ with extension degree $e \geq 3$}{Binary characteristic p=2 with extension degree e >= 3}}

By Theorem \ref{thm: prescribed trace}, there exists some primitive element $\omega \in \mathbb{F}_{2^e}^\ast$ such that $\mathrm{Tr}(1/\omega) = 0$. Recall from Definitions \ref{def: pseudo residues over binary extensions}, \ref{def: restricted parameter regime} that in this parameter regime, we have $\Omega_q = W_q(\omega)$, where
\begin{equation*}
    W_q \defeq W_q(\omega) = \left\lbrace \omega^{2i} : i = 0, 1, \ldots, 2^{e-1} - 2 \right\rbrace
\end{equation*}
consists of all even powers of $\omega$ between 0 and $2^{e}-4$, inclusive. We define the square root function over $W_q$ as follows.

\begin{definition}\label{def: square roots over binary extension}
    Let $q = 2^e$ where $e \geq 3$. Define the square root operation $\sqrt{\cdot}: W_q \to \mathbb{F}_{2^e}^\ast$ as follows: given some integer $0 \leq i \leq 2^{e-1}-2$ and $\omega^{2i} \in W_q$, set $\sqrt{\omega^{2i}} \defeq \omega^i$.
\end{definition}

Furthermore, we set
\begin{equation*}
    a = \omega^{2^{e-1}} \in W_q \quad \text{and} \quad b = \omega \in \mathbb{F}_{2^e}^\ast \setminus W_q.
\end{equation*}
We now show that $\lbrace a, b \rbrace \subseteq B_1(1)$. To do so, we first prove the following lemma.

\begin{lemma}\label{lem: B11 element inverse is trace kernel}
    Fix $\mathbb{F}_{q} = \mathbb{F}_{2^e}$ where $e \geq 3$, and let $y \in \mathbb{F}_q^*$. Then $y \in B_1(1) \setminus \lbrace 0 \rbrace$ if and only if $\mathrm{Tr}(1/y) = 0$.
\end{lemma}

\begin{proof}
    In the forward direction, suppose there exists $m \in \mathbb{F}_q^\ast$ such that $y = m^{-1} + m \neq 0$. Then, moving all terms to one side and multiplying through by $my^{-2}$ yields
    \begin{align*}
        \left( \frac{m}{y} \right)^2 + \left(\frac{m}{y}\right) + \frac{1}{y^2} = 0
    \end{align*}
    which, by Theorem \ref{thm: binary artin schreier}, implies the trace of the constant term $\mathrm{Tr}(y^{-2}) = 0$. In turn, this implies $\mathrm{Tr}(1/y) = 0$, as desired.

    In the reverse direction, suppose $\mathrm{Tr}(1/y)=0$, so that $\mathrm{Tr}(y^{-2}) = 0$. Then by Theorem \ref{thm: binary artin schreier}, there exists some $x \in \mathbb{F}_q^\ast$ such that
    \begin{equation*}
        x^2 + x + \frac{1}{y^2} = 0.
    \end{equation*}
    Multiplying through by $y^2$ yields
    \begin{equation*}
        (xy)^2 + (xy)y + 1 = 0 \quad\implies\quad y = \frac{1}{xy} + xy = m^{-1} + m \in B_1(1)
    \end{equation*}
    where $m \defeq xy$, as desired.
\end{proof}

Note that, given some primitive element $\omega \in \mathbb{F}_{2^e}^\ast$, every element of $y \in B_1(1)\setminus \lbrace 0 \rbrace$ may be written as $y = \omega^{i }$ for some integer $i$. We then rewrite Lemma \ref{lem: B11 element inverse is trace kernel} into the following corollary, which is immediate.

\begin{corollary}\label{cor: B11 element inverse is trace kernel - generator version}
    Fix $\mathbb{F}_{q} = \mathbb{F}_{2^e}$ where $e \geq 3$, and let $\omega\in \mathbb{F}_q^\ast$ be a primitive element;  Then $\omega^{i} \in B_1(1) \setminus \lbrace 0 \rbrace$ if and only if $\mathrm{Tr}(\omega^{-i}) = 0$.
\end{corollary}

The following lemma gives a binary extension field analog of Lemma \ref{lem: B11 has one QR and one non QR}.

\begin{lemma}
    Fix $\mathbb{F}_{q} = \mathbb{F}_{2^e}$ where $e \geq 3$, and let $\omega\in \mathbb{F}_q^\ast$ be a primitive element satisfying $\mathrm{Tr}(\omega^{-1}) = 0$; such $\omega$ exists by Theorem \ref{thm: prescribed trace}. Set
    \begin{equation*}
        a = \omega^{2^{e-1}} \in W_q \quad \text{and} \quad b = \omega \in \mathbb{F}_q^\ast \setminus W_q.
    \end{equation*}
    Then $\lbrace a, b \rbrace \subseteq B_1(1)$.
\end{lemma}

\begin{proof}
    By assumption, $\mathrm{Tr}(\omega^{-1}) = 0$; by Corollary \ref{cor: B11 element inverse is trace kernel - generator version}, it follows that $b = \omega \in B_1(1)$. Next, observe that
    \begin{equation*}
        \mathrm{Tr}\left(\omega^{-2^{e-1}}\right) = \mathrm{Tr}\left(\left(\omega^{-1}\right)^{2^{e-1}} \right) = 0.
    \end{equation*}
    It follows from Corollary \ref{cor: B11 element inverse is trace kernel - generator version} that $a = \omega^{2^{e-1}} \in B_1(1)$, completing the proof.
\end{proof}

\begin{observation}\label{obs: scale pairs over binary extension}
    Let $q = 2^e$ for extension degree $e\geq 3$ and let $\omega\in \mathbb{F}_q^\ast$ be a primitive element satisfying $\mathrm{Tr}(\omega^{-1}) = 0$; such $\omega$ exists by Theorem \ref{thm: prescribed trace}. For $\gamma \in W_q$, define $\sqrt{\gamma}$ as in Definition \ref{def: square roots over binary extension}. Fix
    \begin{equation*}
    a = \omega^{2^{e-1}} \in W_q \quad \text{and} \quad b = \omega \in \mathbb{F}_{2^e}^\ast \setminus W_q.
    \end{equation*}
    Then for all $\delta \in W_q$,
    \begin{equation}
    \left\lvert \bigcup_{\gamma \in W_q \setminus \lbrace \delta \rbrace} \sqrt{\gamma}\cdot \lbrace a, b \rbrace \right\rvert = q-4.
\end{equation}
\end{observation}

\begin{proof}
    First, we show that for distinct $\gamma, \gamma' \in W_q$, we have
    \begin{equation*}
        \sqrt{\gamma}\cdot \lbrace a, b \rbrace \cap \sqrt{\gamma'}\cdot \lbrace a,b\rbrace = \varnothing.
    \end{equation*}
    Indeed, without loss of generality, let $\gamma = \omega^{2i}$ and $\gamma' = \omega^{2j}$, where $0 \leq i < j \leq 2^{e-1}-2$. Then
    \begin{equation*}
        \sqrt{\gamma} \cdot \lbrace a, b \rbrace= \lbrace \omega^{1 + i}, \omega^{2^{e-1} + i} \rbrace \quad\text{and} \quad \sqrt{\gamma'} \cdot \lbrace a, b \rbrace = \lbrace \omega^{1 + j}, \omega^{2^{e-1} + j} \rbrace
    \end{equation*}
    are disjoint, since $1+i \neq 1+j$ and $2^{e-1} + i \neq 2^{e-1} + j$; additionally, the intervals $$1 \leq 1+i < 1 + j \leq 2^{e-1} - 1$$ and $$2^{e-1} \leq 2^{e-1} + i < 2^{e-1} + j \leq 2^{e} - 2$$ have no overlap. Next, let $\delta = \omega^{2d} \in W_q$ for some integer $0 \leq d \leq 2^{e-1}-2$. Then we have
    \begin{align*}
    \left\lvert \bigcup_{\gamma \in W_q \setminus \lbrace \delta \rbrace} \sqrt{\gamma}\cdot \lbrace a, b \rbrace \right\rvert &=\left\lvert \bigsqcup_{\substack{0 \leq i \leq 2^{e-1}-2\\ i \neq d \;\;\;\;\;}} \omega^i\cdot \left\lbrace \omega, \omega^{2^{e-1}} \right\rbrace \right\rvert\\
    &= \sum_{\substack{0 \leq i \leq 2^{e-1}-2\\ i \neq d \;\;\;\;\;}} \left\lvert \left\lbrace \omega^{1+i}, \omega^{2^{e-1}+i} \right\rbrace \right\rvert\\
    &= (2^{e-1} -2)(2)\\
    &= q - 4
    \end{align*}
    as desired.
\end{proof}

\section{Proof of Main Theorem}\label{sec: bounding subcase}

    \subsection{mQM: ``mini QM''}
    As alluded to previously, we now consider a subcase of QM, subject to the restrictions set forth in Definition \ref{def: restricted parameter regime}. The goal of this subsection is to establish the notation we will use for this subcase, and adjust relevant notions to reflect its restricted parameter regime.

\begin{definition}[mQM]\label{def: mQM}
    We denote by $\mathrm{mQM}(\pmb{\alpha},\mathbf{T},\mathbf{b})$ an instance of $\mathrm{QM}(\pmb{\alpha},\mathbf{T},\mathbf{b})$ (Algorithm \ref{alg: QM Full}), subject to the additional constraints of Definition \ref{def: restricted parameter regime}.
\end{definition}

In particular, Definition \ref{def: restricted parameter regime} restricts us to only considering lines $f \in \mathbb{F}_q[x]$ satisfying $g(\mathbf{f}) \in \Omega_q$; accordingly, we need only distinguish among leakage transcripts corresponding to lines with quadratic residue coefficient products.

\begin{definition}[Leakage transcript validity]\label{def: valid leakage transcript}
    Given some $\pmb{\alpha} \in \Omega_q^t$ and fixed sequence of leakage sets $\mathbf{T}= (T_1, \ldots, T_t)$, we say that a $t$-bit leakage transcript $\mathbf{b} \in \lbrace 0, 1\rbrace^t$ is valid if there exists some $f(x) \in \mathbb{F}_q [x]$, $\deg(f) \leq 1$ 
    satisfying $g_{0,1}(\mathbf{f}) \in \Omega_q$ such that 
    \begin{equation}
        \mathbf{b}_i = \begin{cases}
            0 & f(\alpha_i) \in T_{j}\\
            1 & \text{else}
        \end{cases}
    \end{equation}
    for all $i \in [t]$.
\end{definition}

Since $\mathrm{mQM}$ need only distinguish between coefficient products lying in $\Omega_q$ rather than all of $\mathbb{F}_q$, we adjust the notion of success accordingly.

\begin{definition}[$\mathrm{mQM}$ validity]\label{def: mQM validity}
    We say that an instance of $\mathrm{mQM}$ is $(s,t)$-valid if for any choice of size $s \in \mathbb{Z}^+$ sized subset $S \subseteq \Omega_q$, there exists some fixed $\pmb{\alpha} \in S^t$ and fixed sequence of leakage sets $\mathbf{T}$ such that $\mathrm{mQM}(\pmb{\alpha}, \mathbf{T}, \mathbf{b})$ returns ``\texttt{Success!}'' for every $t$-bit valid leakage transcript $\mathbf{b}$. 
\end{definition}

Finally, we formally define the round complexity of $\mathrm{mQM}$.

\begin{definition}\label{def: mQM round complexity}
    We say that $s$-server $\mathrm{mQM}$ has round complexity $t$ if, for all $(s,t')$-valid $\mathrm{mQM}$ schemes, $t' \geq t$.
\end{definition}

    \subsection{pQM: ``projection QM''}
    We now define an algorithm\footnote{So-called as it ``projects'' subsets $B_\gamma \subseteq \mathbb{F}_q[x]$ onto $B_1(1) \subseteq \mathbb{F}_q$.} $\mathrm{pQM}$, analogous to Algorithm~\ref{alg: QM Full}, presented as Algorithm \ref{alg:pQM}. Our main objective in the section is to show that a valid $\mathrm{mQM}$ scheme implies a successful $\mathrm{pQM}$ scheme, which is established in Theorem \ref{thm: mQM implies pQM}. To that end we define a notion of $\mathrm{pQM}$ validity similar to that of Definition \ref{def: mQM validity}.

\begin{algorithm}

    \caption{$\mathrm{pQM}(\mathbf{V}, \mathbf{b})$}\label{alg:pQM}
    \DontPrintSemicolon
    \SetKwInOut{Input}{Input}\SetKwInOut{Output}{Output}
    \Input{leakage sets $\mathbf{V} = \left(V_1, V_2, \ldots, V_t \right)$}
    \Input{leakage transcript $\mathbf{b} \in \lbrace 0, 1 \rbrace^t$}
    \For{$\gamma \in \Omega_q$}
    {
    $S_\gamma^0 \gets B_1(1)$ \\
    \For{$i \in [t]$}
    {
        $S_\gamma^i \gets \begin{cases} S_\gamma^{i-1} \setminus (1/\sqrt{\gamma}) V_i & \mathbf{b}_i = 0 \\ S_\gamma^{i-1} \setminus (1/\sqrt{\gamma}) (\mathbb{F}_q^* \setminus V_i) & \mathbf{b}_i = 1 \end{cases}$
    } 
    }
    \If{$\left\lvert\left\lbrace \gamma : \gamma \in \Omega_q, \; S_\gamma^t \neq \varnothing \right\rbrace\right\rvert \leq 1$}{
    \Return{\texttt{Success!}}
    }
    \Return{\texttt{Fail!}}\\
    
\end{algorithm}

\begin{definition}
    We say that an instance of $\mathrm{pQM}$ is $t$-valid if there exists some  sequence $\mathbf{V}$ of leakage sets such that $\mathrm{pQM}(\mathbf{V}, \mathbf{b})$ succeeds for all $\mathbf{b} \in \lbrace 0, 1 \rbrace^t$. 
\end{definition}

The round complexity is defined analogously to that Definition \ref{def: mQM round complexity}; we restate it for reference.

\begin{definition}\label{def: pQM round complexity}
    We say that $s$-server $\mathrm{pQM}$ has round complexity $t$ if, for all $t'$-valid $\mathrm{pQM}$ schemes, $t' \geq t$.
\end{definition}

We may now formally state the key theorem that will allow us to analyze the round complexity of $\mathrm{pQM}$ to derive a lower bound on that of $\mathrm{mQM}$ - and by extension, $\mathrm{QM}$ itself.

\begin{theorem}\label{thm: mQM implies pQM}
    Let $s,t\in \mathbb{Z}^+$ and suppose there exists an $(s,t)$-valid instance of $\mathrm{mQM}$. Then there exists an $r$-valid instance of $\mathrm{pQM}$, for some positive integer $r \leq t$.
\end{theorem}

\begin{proof}
    Fix some server schedule $\pmb{\alpha} \in \Omega_q^t$ supported on $s$ distinct server labels in $\Omega_q$. Suppose that $\mathbf{T}$ is a sequence of leakage sets defining an $(s,t)$-valid mQM scheme; then for any valid $t$-bit leakage transcript $\mathbf{b}$, $\mathrm{mQM}(\pmb{\alpha}, \mathbf{T}, \mathbf{b})$ returns ``\texttt{Success!}'', implying that at the end of its execution, $\mathbf{B}(\gamma) = \varnothing$ for all but one value $\gamma \in \Omega_q$; denote this unique value $\delta \in \Omega_q$.

    Let $\gamma \in \Omega_q\setminus\lbrace \delta\rbrace$. Then for each $h_m^\gamma \in B_\gamma$ there exists some round $i_m \in [t]$ such that precisely one of
    \begin{equation}\label{eqn: fwd equiv b = 0}
        \mathbf{b}_{i_m} = 0 \quad \text{and} \quad h_m^\gamma \left( \alpha_{i_m} \right) \in T_{i_m}
    \end{equation}
    or
    \begin{equation}\label{eqn: fwd equiv b = 1}
        \mathbf{b}_{i_m} = 1 \quad \text{and} \quad h_m^\gamma \left( \alpha_{i_m} \right) \in \mathbb{F}_q \setminus T_{i_m}
    \end{equation}
    holds, otherwise $h_m^\gamma$ would never be removed from $B_\gamma$ and $\mathbf{B}(\gamma) \neq \varnothing$; take $i_m$ to be the least such value. Furthermore, assume that $\mathbf{b}_{i_m} = 0$ so that \eqref{eqn: fwd equiv b = 0} holds; the case $\mathbf{b}_{i_m} = 1$ of \eqref{eqn: fwd equiv b = 1} is verbatim. 
    
    From $\mathbf{T}$, construct the length $t$ sequence $\mathbf{V}$ of leakage sets by applying Definition \ref{def: converted eliminators} to each entry of $T_{i} \mapsto V_{i}$, $i \in [t]$. Then by Theorem \ref{thm: eliminator conversion}, \eqref{eqn: fwd equiv b = 0} holding in some round $i_m \in [t]$ implies
    \begin{equation}\label{eqn: fwd sync removal}
        h_m^\gamma \left( \alpha_{i_m} \right) \in T_{i_m}\quad \implies \quad g_m(1) \in \frac{1}{\sqrt{\gamma}} V_{i_m}.
    \end{equation}
    Consider an instance $\mathrm{pQM}(\mathbf{V}, \mathbf{b})$. By \eqref{eqn: fwd sync removal}, $h_m^\gamma \not\in \mathbf{B}(\gamma)$ implies $g_m(1) \not\in S_\gamma^{j}$ at the conclusion of every round $j\in [t]$, $j \geq i_m$. Since every $h_m^\gamma$, $m \in \mathbb{F}_q^\ast$ is removed in some round $i_m \in [t]$ of $\mathrm{mQM}(\pmb{\alpha}, \mathbf{T}, \mathbf{b})$, it follows that every $g_m(1)$, $m \in \mathbb{F}_q^\ast$ will be removed no later than the same round $i_m \in [t]$ of $\mathrm{pQM}(\mathbf{V}, \mathbf{b})$. Hence,
    \begin{equation}\label{eqn: fwd equiv same 0 stats}
        \mathbf{B}(\gamma) = \varnothing \quad \implies \quad S_\gamma^t = \varnothing.
    \end{equation}
    Since Equation \ref{eqn: fwd equiv same 0 stats} holds for all $\gamma \in \Omega_q \setminus \lbrace \delta \rbrace$, it follows that there exists \textit{at most} one value, namely $\delta \in \Omega_q$, satisfying $S_\delta^t \neq \varnothing$; accordingly, $\mathrm{pQM}$ will return ``\texttt{Success!}'', as desired, in at most $t$ rounds.
\end{proof}

\subsection{Round Complexity of pQM, and Proof of Theorem~\ref{thm: full-case bound}}\label{sec:mainpf}

Recall Definitions \ref{def: mQM round complexity}, \ref{def: pQM round complexity}. Theorem \ref{thm: mQM implies pQM} shows that any under bound on the round complexity of pQM (Algorithm \ref{alg:pQM}) is an under bound on the round complexity of mQM, which in turn is a sub-case of QM (Algorithm \ref{alg: QM Full}). With this relation established, we now focus on the task of under bounding the round complexity of pQM. We will show the following theorem.

\begin{theorem}\label{thm: sub-case bound}
    Fix $\mathbb{F}_q = \mathbb{F}_{p^e}$ for prime $p > 2$ and $q \neq 5$; or $p = 2$ and $e \geq 3$. Then any $t$-valid pQM scheme over $\mathbb{F}_q$ must satisfy
    \begin{equation*}
        t \geq \begin{cases}
            2\log_2(q-1) - 3 & p>2, \; p^e \neq 5\\
            2\log_2(q-2) - 4 & p = 2, e \geq 3.
        \end{cases}
    \end{equation*}
\end{theorem}

We will prove Theorem \ref{thm: sub-case bound} later, as it relies on Theorem \ref{thm: output list size 2} in the next section. 
For now, we observe that Theorem~\ref{thm: sub-case bound}, along with the above reasoning, implies \cref{thm: full-case bound}. We will do so in two steps: 

\begin{enumerate}
    \item Theorem \ref{thm: full-case bound restricted to k=2} shows that Theorem \ref{thm: mQM implies pQM} and Theorem \ref{thm: sub-case bound} imply Theorem \ref{thm: full-case bound} in the restricted parameter regime where $k=2$, $i=0$, and $j=1$ - i.e., under Assumption~\ref{assm:k2}.

    \item  shows that Theorem \ref{thm: full-case bound restricted to k=2} implies Theorem \ref{thm: full-case bound} for all $k \geq 2$, via a reduction argument.
\end{enumerate}

\begin{theorem}[Theorem \ref{thm: full-case bound}, restricted to $k=2$]\label{thm: full-case bound restricted to k=2}
    Fix $\mathbb{F}_q = \mathbb{F}_{p^e}$; let $s \geq 3$, $k = 2$, and suppose there exists a $t$-bit, $s$-server QM scheme for $g_{0,1}$ and RS codes of dimension $k$ (Definition \ref{def: QM}) over $\mathbb{F}_q$. Then its download bandwidth must satisfy 
    \begin{equation*}
        t \geq \begin{cases}
            2\log_2(q-1) - 3 & p > 2\\
            2\log_2(q-2) - 4 & p = 2.
        \end{cases}
    \end{equation*}
\end{theorem}

\begin{proof}
First, we argue that any lower bound on the bandwidth of mQMs implies the same lower bound for QMs.  To do this, we must argue that the restrictions in 
    Definition~\ref{def: restricted parameter regime} are without loss of generality.  Recall that, once we have assumed that $k=2$, Definition~\ref{def: restricted parameter regime} restricts the mQM coefficient product $g(\mathbf{f})\in \Omega_q$ and requires all queries to be made to servers indexed by some $\alpha \in \Omega_q$. The former restriction is strictly a sub-problem of QM. The latter restriction is without loss of generality, since the number $s$ of queried servers is at most $|\Omega_q|$. Indeed, the client may query any $s$ servers but must download \emph{at least} one bit per query, so the bound of Theorem \ref{thm: full-case bound} is relevant only when $s \leq 2 \log_2(q-2) - 4$ in the characteristic 2 case, or $s \leq 2 \log_2(q-1) - 3$ in the odd characteristic case. On the other hand, the order of $\Omega_q$ is at least $\left\lvert \Omega_q \right\rvert \geq (q-2)/2$, which is always greater than $2 \log_2(q-1) - 3$ for all prime powers $q$.
    
    Now Theorems \ref{thm: mQM implies pQM} and \ref{thm: sub-case bound}, along with the fact noted above that lower bounds for mQMs imply the same lower bounds for QMs, implies Theorem~\ref{thm: full-case bound} in all cases except when $q=2, 4, 5$.  However, for such values of $q$, Theorem \ref{thm: full-case bound} is vacuously true: it says that download bandwidth $t$ satisfies $t \geq 2 \log_2(5-1) - 3 = 1$, which is less than the $\log_2(5) >2$ bits needed to represent an element of $\mathbb{F}_5$. A similar phenomenon holds for $q = 2, 4$.
\end{proof}

\fullcasebound*

\begin{proof}
    Let $\mathcal{B}(p)$ denote the right-side term of the inequality in \eqref{eqn: p case bound}. For $k > 2$, distinct $i, j \in [0,k-1]$ with $i< j$, $\tau \in \mathbb{Z}^+$, and $s\geq 3$, suppose exists a $\tau$-bit, $s$-server QM $\Phi$ over $\mathbb{F}_q = \mathbb{F}_{p^e}$ computing $g_{i,j}(\mathbf{f}) = f_i f_j$ for all $\mathbf{f} \in \mathbb{F}_q^k \simeq \mathrm{RS}_q[n,k]$. Consider the dimension $2$ subspace of $\mathbb{F}_q^k$ given by
    \begin{equation*}
        \mathcal{C} = \left\lbrace \mathbf{f} = f_i e_i + f_j e_j : f_i, f_j \in \mathbb{F}_q \right\rbrace \simeq \mathrm{RS}_q[n,2]
    \end{equation*}
    where $e_i, e_j$ denote the $i$th, $j$th standard basis vectors, respectively. In words, we may view $\mathcal{C}$ as the set of all Reed-Solomon message polynomial coefficient vectors which are 0 everywhere except at coordinates $i,j$, where they may take arbitrary values $f_i, f_j \in \mathbb{F}_q$. 
    
    Observe that since $\mathcal{C}\subseteq \mathbb{F}_q^k$, we know $\Phi$ is trivially a $\tau$-bit, $s$-server QM for $\mathcal{C}$. For any $\mathbf{f} \in \mathcal{C}$, each server indexed by $\alpha \in \mathbb{F}_q^\ast$ locally holds the codeword symbol $f(\alpha) = f_i \alpha^i + f_j \alpha^j$. Each server may also locally scale their codeword symbol by $\alpha^{-i}$ to recover $\hat{f}(\alpha) \defeq \alpha^{-i} f(\alpha) = f_i + f_j \alpha^{j-i}$. 
    
    Thus, up to local rescaling, each server indexed by some $\alpha \in \mathbb{F}_q^\ast$ holds an evaluation point of the form $\hat{f}(\alpha) = b + m \alpha^r$ for some $r \defeq j-i \in [1, k-1]$ and $m,b \in \mathbb{F}_q$. Each server may then locally re-index by $\alpha^r \in \mathbb{F}_q^\ast$; after such a re-indexing, each server labeled by some $\beta = \alpha^r \in \mathbb{F}_q^\ast$ holds a \emph{linear} evaluation $h(\beta) = b + m \beta = \hat{f}(\alpha)$. Let $S'$ be the set of server labels $\beta \in \mathbb{F}_q^\ast$ and $s' = |S'| \leq n$; then $(h(\beta))_{\beta \in S}$ is a Reed-Solomon codeword in $\mathcal{C}' \defeq \mathrm{RS}_q[s' \leq n, 2]$. Importantly, $\mathcal{C}'$ need not be full-length\footnote{If $r$ divides $q-1$, not all evaluation points may be available; that is, only certain $\beta \in \mathbb{F}_q^\ast$ may appear, corresponding to the image of the map $x \mapsto x^r$ over $\mathbb{F}_q^*$.} over $\mathbb{F}_q$. Nonetheless, any bandwidth lower bound for QM over a full-length RS code must hold for QM over any subset of codeword positions; otherwise, a client could lower download bandwidth in the former case by querying a subset of available servers. 

    Thus, recovering $g_{0,1}(\mathbf{h}) = mb$ is equivalent to recovering $g_{i,j}(\mathbf{f})$ for $\mathbf{f} \in \mathcal{C} \subseteq \mathbb{F}_q^k$. Since Theorem \ref{thm: full-case bound restricted to k=2} implies that any QM scheme for full-length RS codes recovering $g_{0,1}(\mathbf{h}) = mb$ must download at least $\mathcal{B}(p)$ bits, the same requirement holds for our RS code after the relabeling $\beta = \alpha^r$ of evaluation points. We see that $\Phi$ must incur a download at least $\tau \geq \mathcal{B}(p)$, as desired.
\end{proof}

\subsubsection{Size of pQM final states}

We now focus exclusively on pQM as presented in Algorithm \ref{alg:pQM}, placing aside the problem of coefficient product recovery altogether. We show that, if pQM succeeds, then it must have removed nearly all points from consideration. We state this more precisely.

\begin{theorem}\label{thm: output list size 2}
    Fix $\mathbb{F}_q = \mathbb{F}_{p^e}$ for prime $p > 2$ and $q \neq 5$. Let $\mathbf{V}$ be some sequence of leakage sets $V_i$, $i \in [t]$ defining an $t$-valid $\mathrm{pQM}$ (Algorithm \ref{alg:pQM}). 
    
    Suppose for some $\mathbf{b}\in \lbrace 0, 1 \rbrace^t$ there exists a unique $\delta \in \Omega_q$ such that $S_\delta^t \neq \varnothing$ at the conclusion of $\mathrm{pQM}(\mathbf{V}, \mathbf{b})$. Then $$|S_\delta^t| \leq \begin{cases}
             2 & p \text{ odd }\\
             3 & p = 2, e \geq 3.
         \end{cases}$$
\end{theorem}

\begin{proof}
    Assume $\mathrm{pQM}$ is $t$-valid and suppose that at the conclusion of an instance thereof we have
    \begin{equation}\label{eqn: only one survivor}
        S_\delta^t \neq \varnothing \quad \text{and} \quad S_\gamma^t = \varnothing \;\forall \; \gamma \in \Omega_q\setminus \lbrace \delta \rbrace.
   \end{equation}
    Recall that $\mathrm{pQM}$ initializes $S_\gamma^0 = B_1(1)$, and that any final state $S_\sigma^t$, $\sigma \in \Omega_q$ is simply a pruning of $S_\sigma^0$ via iterated set-minus operations. We may assume without loss of generality that $\mathbf{b}=0^t$ (if $\mathbf{b}_i =1$ for some $i \in [t]$, replace $V_i$ with its set complement). Accordingly, note that
    \begin{equation}
        S_\sigma^j = B_1(1) \setminus \frac{1}{\sqrt{\sigma}}\left( \bigcup_{i=1}^j V_i \right)
    \end{equation}
    for all $\sigma \in \Omega_q$, $j \in [t]$. Then $S_\gamma^t = \varnothing$ implies
    \begin{equation}
        \sqrt{\gamma} \cdot B_1(1) \subseteq \bigcup_{i=1}^t V_i
    \end{equation}
    for all $\gamma \in \Omega_q\setminus \lbrace \delta \rbrace$. By Theorem \ref{thm: size of scaled pairs set}, there exists a pair of non-zero field elements $a, b \in B_1(1)\setminus \lbrace 0 \rbrace$, with $a \neq b$, such that $a \in \Omega_q$ and $b \in \mathbb{F}_q^\ast \setminus \Omega_q$ such that
    \begin{equation}
        \sqrt{\gamma} \cdot \lbrace a,b \rbrace \subseteq \sqrt{\gamma}\cdot B_1(1) \subseteq \bigcup_{i=1}^t V_i
    \end{equation}
    for all $\gamma \in \Omega_q\setminus \lbrace \delta \rbrace$. Additionally, Theorem \ref{thm: size of scaled pairs set} gives
    \begin{equation}
        \left\lvert \bigcup_{i=1}^t V_i \right\rvert \geq \left\lvert \bigcup_{\gamma \in \Omega_q \setminus \{ \delta\} }\sqrt{\gamma} \cdot \{a,b \} \right\rvert \geq \begin{cases}
             q-3 & p \text{ odd }\\
             q-4 & p = 2, e \geq 3
         \end{cases}
    \end{equation}
    Finally, since $S_\delta^t \subseteq \mathbb{F}_q^\ast$, it must be that
    \begin{equation}
        \left\lvert S_\delta^t \right\rvert = \left\lvert B_1(1) \setminus \frac{1}{\sqrt{\delta}}\left( \bigcup_{i=1}^t V_i \right) \right\rvert \leq \left\lvert \mathbb{F}_q^\ast \setminus \left( \bigcup_{i=1}^t V_i \right) \right\rvert \leq \begin{cases}
             2 & p \text{ odd }\\
             3 & p = 2, e \geq 3
         \end{cases}
    \end{equation}
    as desired.
\end{proof}

The following observation is an immediate corollary of Theorem \ref{thm: output list size 2}; in particular, it restates Theorem \ref{thm: output list size 2} in terms of \eqref{eqn: only one survivor}.

\begin{corollary}\label{cor: if one survives, only has 2 things in it}
    Assume the hypotheses and notation of Theorem \ref{thm: output list size 2}; then $\mathrm{pQM}(\mathbf{V}, \mathbf{b})$ succeeds if and only if
    \begin{equation}
        \left\lvert\bigcup_{\sigma \in \Omega_q} S_\sigma^t \right\rvert \leq \begin{cases}
             2 & p \text{ odd }\\
             3 & p = 2, e \geq 3
         \end{cases}.
    \end{equation}
\end{corollary}

\subsubsection{Lower-bounding Round Complexity}\label{subsubsec: adversarial pQM}

We conclude this section with the proof of Theorem \ref{thm: sub-case bound}. We construct a game that is equivalent to pQM, but in an ``adversarial'' context. The intuition is as follows: a player - say, \textit{Alice} - specifies some query schedule of servers $\mathbf{u} \in \Omega_q^t$ whose support is over $s \geq 3$ distinct server labels. Then during any round $i \in [t]$, Alice specifies some leakage set $V_i$, and an adversary - say, \textit{Eve} -chooses the leakage bit $\mathbf{b}_i \in \lbrace 0, 1 \rbrace$ that removes the least (i.e., at most half) of the points from consideration. Since pQM does not conclude if there are more than 2 points still in consideration (Theorem \ref{thm: output list size 2}, Corollary \ref{cor: if one survives, only has 2 things in it}), we can see that the round complexity must be $O(2 \log_2(q))$. The following is nothing more than a formalization of this intuition.

\begin{proof}[Proof of Theorem \ref{thm: sub-case bound}]
    Consider a round-based game played between Alice and an (adversarial but honest) Eve.

    \paragraph{Initialization:} For all $\gamma \in \Omega_q$, Alice initializes\footnote{This is analogous to the initialization in Algorithm \ref{alg:pQM}.} the sets $S_\gamma^0 = B_1(1)$.

    \paragraph{Rounds:} For each round $i \in [t]$,
    \begin{enumerate}
        \item Alice chooses an arbitrary $V_i \subseteq \mathbb{F}_q$, and denotes $Y_0 \defeq V_i$, $Y_1 \defeq \mathbb{F}_q^\ast \setminus V_i$.

        \item Eve computes
        \begin{equation}
            \mathbf{b}_i \defeq \argmax_{b \in \lbrace 0, 1 \rbrace}\left( \sum_{\gamma \in \Omega_q} \left\lvert S_\gamma^{i-1} \setminus \left(\frac{1}{\sqrt{\gamma}} Y_b\right) \right\rvert \right)
        \end{equation}

        \item Alice updates the sets $S^{i-1}_\gamma$, $\gamma \in \Omega_q$ according to Eve's choice of $\mathbf{b}_i$; concretely, for all $\gamma \in \Omega_q$, Alice sets
        \begin{equation}
            S_\gamma^i \defeq S_\gamma^{i-1} \setminus \left(\frac{1}{\sqrt{\gamma}} Y_{\mathbf{b}_i} \right).
        \end{equation}
    \end{enumerate}

    \paragraph{Conclusion:} Following the final round $i=t$, Alice wins if there exists at most one $\delta \in \Omega_q$ such that $S_\delta^t \neq \varnothing$; otherwise, Eve wins. By observation, the above game is equivalent to pQM where the leakage set $V_i$ and leakage bit $\mathbf{b}_i$ are determined round-to-round by Alice and Eve, respectively. Eve's choice of $\mathbf{b}_i$ ensures that at most half of all points still under consideration are removed by Alice's updates in a given round $i \in [t]$, representing a worst-case leakage transcript for $\mathrm{pQM}$. Furthermore, Eve is not concerned by which sets $S_\gamma^i$ may be empty or not at any given round; by Corollary \ref{cor: if one survives, only has 2 things in it} Eve knows that if there are more than two or three points (among all $S_\gamma^i$, $\gamma \in \Omega_q$) still under consideration at any round, then Alice cannot have won. More precisely, if Alice wishes to win in $t$ rounds, then $t$ must satisfy:
    
    \begin{itemize}
        \item \textbf{If $q = p^e$ for an odd prime $p$:} then $\Omega_q = \mathrm{QR}_q$; by Corollary \ref{cor: if one survives, only has 2 things in it} and Observation \ref{obs: size of B11} we have
        \begin{align}
        t &\geq \left\lceil \log_2 \left( \sum_{\gamma \in \mathrm{QR}_q} \left\lvert S_\gamma^0 \right\rvert \right) - \log_2(2)\right\rceil \\ \nonumber \\
        & = \left\lceil\log_2\left( \frac{(q-1)(q+1)}{4} \right) \right\rceil- 1 \\
        & \geq \log_2((q-1)^2) - 3.
        \end{align}

        \item \textbf{If $q = 2^e$, for $e \geq 3$:} then $\Omega_q = W_q$; by Corollary \ref{cor: if one survives, only has 2 things in it} and Observation \ref{obs: size of B11} we have
        \begin{align}
        t &\geq \left\lceil \log_2 \left( \sum_{\gamma \in W_q} \left\lvert S_\gamma^0 \right\rvert \right) - \log_2(3)\right\rceil \\ \nonumber \\
        & \geq \left\lceil\log_2\left(  \frac{\left( 2^{e} - 2 \right) \left(2^e\right)}{4} \right) \right\rceil- 2 \\
        & \geq \log_2((q-2)^2) - 4.
        \end{align}
    \end{itemize}
\end{proof}

\section{Necessity of Non-linear Local Computation}\label{sec: non-linear necessity}
In this section, we consider the setting where $k=2$ and $i=0, j=1$. We show that in this case, an even stronger bound must hold if the leakage functions are $\mathbb{F}_p$-linear. Though this result is not invoked in the proof of our main result (Theorem \ref{thm: full-case bound}), it motivates the consideration of arbitrary (non-linear) leakage functions as in Definition \ref{def: leakage function}. This is an interesting contrast to the bandwidth savings which can be achieved by linear schemes for computing \emph{linear} functions (e.g., \cite{GW16}, \cite{TYB17}, \cite{SW21}, \cite{Kiah24}). The following result shows that linear schemes cannot perform \textit{even a single bit} better than the naive download bandwidth for quadratic functions.

\begin{theorem}\label{thm: linear eval impossible}
    Let $\mathcal{C}$ be a Reed-Solomon code over alphabet $\mathbb{F}_q = \mathbb{F}_{p^e}$ of length $n = q-1$ and dimension $k\geq 2$, with evaluation points given by $\mathbb{F}_q^\ast$. Let $t = 2\log_p(q) - 1 = 2e - 1 $ and $i,j \in [0,k-1]$, $i\neq j$. Then for any fixed choice of (not necessarily distinct) $\alpha_1, \alpha_2, \ldots, \alpha_t \in \mathbb{F}_q^\ast$ and $\gamma_1, \gamma_2, \ldots, \gamma_t \in \mathbb{F}_q$, there does not exist any $R \in \mathbb{F}_q[X_1, \ldots, X_t]$ such that
    \begin{equation}
        R\left( \mathrm{Tr}\left( \gamma_i f(\alpha_i)\right) \; : \; i \in [t] \right) = g_{i,j}(\mathbf{f})
    \end{equation}
    for all $f \in \mathbb{F}[x]$, $\deg(f) < k$.
\end{theorem}

We defer the proof of Theorem \ref{thm: linear eval impossible} to the end of this section, after establishing some helpful intermediate observations. 

Intuitively, Theorem \ref{thm: linear eval impossible} says that if servers are restricted to linear local computation, then any strategy must download two field symbols' ``worth'' of bits; in particular, when $k=2$, the user cannot hope to save any download bandwidth over the naive strategy. Note that Theorem \ref{thm: linear eval impossible} has the same setup as QM, except that servers are restricted to evaluating linear functions on their locally held Reed-Solomon codeword symbol; any such linear function may be represented as $x \mapsto \mathrm{Tr}(\gamma x)$ for some choice of $\gamma \in \mathbb{F}_q$. Furthermore, a user queries up to $t$ servers and downloads an element\footnote{We must consider $\mathbb{F}_p$-valued leakage functions if restricting leakage functions to be linear. This overlaps precisely with Definition \ref{def: leakage function} when $p=2$ and $\mathbb{F}_q$ is a binary extension field of degree $e$.} of $\mathbb{F}_p$ per query. Then Theorem \ref{thm: linear eval impossible} says that, for any non-trivial\footnote{Restricting $t \leq 2e-1$ is necessary; if $k=2$ and $t=2e$, then the user can query two servers holding $f(\alpha_1)$ and $ f(\alpha_2)$, respectively, for $e$ symbols each (say, $\mathrm{Tr}(\gamma_i x)$ for some basis $\gamma_i , i \in [e]$ of $\mathbb{F}_q$ over $\mathbb{F}_p$) to download their entire contents. Finding $\prod \mathbf{f}$ is then possible by recovering $f(x)$ entirely.} value of $t \leq 2e-1$, the user cannot deterministically recover the coefficient product of a line $f(x)$.  

\begin{lemma}\label{lem: linear eval impossibility - decomp}
    Let $u, v \in \mathbb{F}_q$ such that at least one of $u,v \neq 0$. Then there exists $m, m', b, b' \in \mathbb{F}_q$ such that $u = m+m'$, $v = b+b'$, and $mb \neq m'b'$.
\end{lemma}

\begin{proof}
    Assume without loss of generality that $u \neq 0$. Assume towards a contradiction that, for all $m,m',b,b' \in \mathbb{F}_q$ satisfying $u = m + m'$ and $v = b+b'$, it is the case that $mb = m'b'$. Consider the following cases:

    \paragraph{$v=0$:} Then $m' = u - m$ and $b' = -b$. It follows that $mb = (u-m)(-b) = -ub + mb$ for all $b \in \mathbb{F}_q$, which implies $-ub = 0$, requiring $u = 0$ - a contradiction.

    \paragraph{$v \neq 0$:} Then $m' = u - m$ and $b' = v-b$. It follows that
    \begin{equation}
        m'b' = (u-m)(v-b) = uv - ub - vm + mb = mb
    \end{equation}
    which requires $uv - ub - vm = 0$ for all $m,b\in \mathbb{F}_q$. Setting $m,b = 0$, this implies $uv = 0$; this is a contradiction, since we assumed that both $u, v \neq 0$.
\end{proof}

\begin{lemma}\label{lem: linear eval impossibility - zero trace}
    Let $t = 2\left[\mathbb{F}_q : \mathbb{F}_p \right] - 1= 2e-1$. For any fixed choice of (not necessarily distinct) $\alpha_1, \ldots, \alpha_t \in \mathbb{F}_q^\ast$ and $\gamma_1, \ldots, \gamma_t \in \mathbb{F}_q$, there exists a choice of $u,v \in \mathbb{F}_q$ not both zero, defining $h(x) \defeq ux + v$, such that $\mathrm{Tr}\left( \gamma_i h(\alpha_i) \right) = 0 \; \forall \; i \in [t]$.
\end{lemma}

\begin{proof}
    For any basis $\lbrace \beta_1, \ldots, \beta_e \rbrace$ of a degree-$e$ extension $\mathbb{F}_q$ over $\mathbb{F}_p$, let $\psi:\mathbb{F}_q \to \mathbb{F}_p^e$ denote the map interpreting $\mathbb{F}_q$ as a vector space over $\mathbb{F}_p$ according to the chosen basis. Furthermore, let $\phi: \mathbb{F}_q \hookrightarrow \mathbb{F}_p^{e\times e}$ be the corresponding embedding of $\mathbb{F}_q$ into $\mathrm{GL}(p,e)$ which maps an element $a \in \mathbb{F}_q$ to the matrix representing multiplication by $a$.
    
    The map $x \mapsto x^{p^w}$ is a $\mathbb{F}_p$-linear map in the coordinates of $\psi(x)$, so there exists some $P \in \mathrm{GL}(p,e)$ such that $P^w \psi(x) = \psi(x^{p^w})$ for all $x \in \mathbb{F}_q$. Then
    \begin{equation}
        \mathrm{Tr}\left( \gamma_i f(\alpha_i) \right) = \sum_{w=0}^{e-1} \left(\gamma_i \alpha_i m \right)^{p^w} + \left(  \gamma_i b\right)^{p^w} = 0, \; i \in [t]
    \end{equation}
    may be equivalently represented as
    \begin{equation}
        \sum_{w=0}^{e-1} \phi\left(\gamma_i^{p^w} \alpha_i^{p^w} \right) P^{w} \psi(m) + \phi\left(\gamma_i^{p^w} \right) P^w \psi(b) = \mathbf{0}, \; i\in [t]
    \end{equation}
    which is a system of $t=2e-1$ linear constraints on the coordinates of $\psi(m), \psi(b)$, of which there are $2e$; thus, the map has a non-trivial kernel.
\end{proof}

\begin{lemma}\label{lem: linear eval impossibility - transcript collision}
    Let $t = 2e-1$. For any fixed choice of $\alpha_1, \ldots, \alpha_t \in \mathbb{F}_q^\ast$ and $\gamma_1, \ldots, \gamma_t \in \mathbb{F}_q$, there exist $f, \ell \in \mathbb{F}_q[x]$ with $\deg(f), \deg(\ell) \leq 1$ and $g_{0,1}(\mathbf{f}) \neq g_{0,1}(\boldsymbol\ell)$ satisfying
    \begin{equation}
        \mathrm{Tr}\left( \gamma_i f\left( \alpha_i \right) \right) = \mathrm{Tr}\left( \gamma_i \ell\left( \alpha_i \right) \right)
    \end{equation}
    for all $i \in [t]$.
\end{lemma}

\begin{proof}
    By Lemma \ref{lem: linear eval impossibility - zero trace}, there exists some non-zero $h(x) = ux + v \in \mathbb{F}_q[x]$ such that 
    \begin{equation}
         \mathrm{Tr}\left( \gamma_i h\left( \alpha_i \right) \right) = 0 \; \forall \;i \in [t].
    \end{equation}
   By Lemma \ref{lem: linear eval impossibility - decomp}, there exists $m, m', b, b' \in \mathbb{F}_q$ such that $u = m+ m'$, $v = b+ b'$, and $mb \neq m'b'$. Define
    \begin{equation}
        f(x) \defeq m x + b \quad \text{and} \quad \ell(x) \defeq -m'x - b'
    \end{equation}
    so that $h(x) = f(x) - \ell(x)$. Then for all $i \in [t]$, we have
    \begin{equation}
        \mathrm{Tr}\left( \gamma_i h\left( \alpha_i \right) \right) = \mathrm{Tr}\left( \gamma_i f\left( \alpha_i \right) \right) - \mathrm{Tr}\left( \gamma_i \ell\left( \alpha_i \right) \right) =0
    \end{equation}
    as desired.
\end{proof}

\begin{lemma}\label{lemma: linear eval impossible k=2}
    Let $\mathcal{C}$ be a Reed-Solomon code over alphabet $\mathbb{F}_q = \mathbb{F}_{p^e}$ of length $n = q-1$ and dimension $k=2$, with evaluation points given by $\mathbb{F}_q^\ast$. Let $t = 2\log_p(q) - 1 = 2e - 1 $ and $i=0, j=1$. Then for any fixed choice of (not necessarily distinct) $\alpha_1, \alpha_2, \ldots, \alpha_t \in \mathbb{F}_q^\ast$ and $\gamma_1, \gamma_2, \ldots, \gamma_t \in \mathbb{F}_q$, there does not exist any $R \in \mathbb{F}_q[X_1, \ldots, X_t]$ such that
    \begin{equation}
        R\left( \mathrm{Tr}\left( \gamma_i f(\alpha_i)\right) \; : \; i \in [t] \right) = g_{0,1}(\mathbf{f})
    \end{equation}
    for all $f \in \mathbb{F}[x]$, $\deg(f) < 2$.
\end{lemma}

\begin{proof}
    Since we require that $R$ deterministically output the coefficient product $g_{0,1}(\mathbf{f})$, we can represent $R$ as a $t$-variate polynomial over $\mathbb{F}_q$ . By Lemma \ref{lem: linear eval impossibility - transcript collision}, there exist a pair of linear functions $f, \ell \in \mathbb{F}_q[x]$ such that
    \begin{equation}
        \mathbf{x} \defeq \left( \mathrm{Tr}\left( \gamma_i f\left( \alpha_i \right) \right) \; : \; i \in [t]\right) = \left( \mathrm{Tr}\left( \gamma_i \ell\left( \alpha_i \right) \right) \; : \; i \in [t]\right) \in \mathbb{F}_p^t
    \end{equation}
    and $g_{0,1}(\mathbf{f}) \neq g_{0,1}(\boldsymbol\ell)$. At least one of $R(\mathbf{x}) \neq g_{0,1}(\mathbf{f})$ or $R(\mathbf{x}) \neq g_{0,1}(\boldsymbol\ell)$ must hold, implying that $R$ has failed to distinguish between their distinct coefficient products.
\end{proof}

Finally, the proof of Theorem \ref{thm: linear eval impossible} follows arguing that the case of $k > 2$ with $i,j \in [0,k-1]$, $i\neq j$ is at least as difficult as the case $k=2$, $i=0$, $j=1$.

\begin{proof}[Proof of Theorem \ref{thm: linear eval impossible}]
    The argument is extremely similar to that of \cref{thm: full-case bound}; we summarize its application to our current setting. As in the proof of \cref{thm: full-case bound}, we may assume that (up to local rescaling), each server indexed by some $\alpha \in \mathbb{F}_q^\ast$ holds an evaluation point of the form $\hat{f} = b + m \alpha^r$, which may be interpreted as the evaluation of a linear polynomial in $\beta \defeq \alpha^r$. If $x \mapsto x^r$ is an permutation on $\mathbb{F}_q^\ast$, then we are in the parameter regime of Lemma \ref{lemma: linear eval impossible k=2}; if not (i.e., $r$ divides $q-1$), then we have strictly fewer evaluation points in the set of linear evaluations, and any lower bound on schemes over full-length Reed Solomon codes must hold for any truncation thereof. Thus the lower bound holds in both cases, as desired.
\end{proof}

\bibliographystyle{alpha}
\bibliography{refs.bib}

\newpage
\appendix

\section{5-bit reconstruction for ``multiplicative'' Shamir secret sharing over GF(7)}\label{apx: LBR example}
Recall that Shamir secret sharing shards a secret $s \in \mathbb{F}_q$ by choosing $c_1, \ldots, c_{k-1} \in \mathbb{F}_q$ uniformly at random to construct the polynomial $h(x) \defeq s + \sum_{i=1}^{k-1} c_i x^i$. To each shareholder indexed by $\alpha \in \mathbb{F}_q$ is dealt a secret share $h(\alpha) \in \mathbb{F}_q$. Any $k$ shareholders can pool their shares and interpolate $h(x)$, recovering $s$ in the process.

We informally recall the notion of ``multiplicative'' Shamir secret sharing first mentioned in Section \ref{sec: related works}. The dealer wishes to shard a secret $s \in \mathbb{F}_q^\ast$, and chooses $c_0, c_1, \ldots, c_{k-1} \in \mathbb{F}_q^\ast$ uniformly at random subject to the constraint that their product is the secret - namely,
\begin{equation}
    \prod_{i=0}^{k-1} c_i = s.
\end{equation}
The dealer constructs the polynomial $h(x) \defeq \sum_{i=0}^{k-1} c_i x^i$, then distributes to each shareholder indexed by $\alpha \in \mathbb{F}_q^\ast$ a secret share $h(\alpha)$. 

Note that in both cases, the shareholders' shares are polynomial evaluations, but only in the former can the secret shares be viewed as a linear function of the secret; in the latter, the secret shares are non-linear in the secret.

We now consider a example of such a ``multiplicative'' secret sharing scheme over $\mathbb{F}_7^\ast$, where $k=2$. As the shares are now non-linear in the secret, we may wonder whether this implies one can do no better than what we might expect as the bandwidth incurred by naive interpolation, $2 \lceil \log_2(|\mathbb{F}_7^\ast|) \rceil = 6$ bits. We briefly show that this is not the case, and in fact we can deterministically recover the secret in 5 bits.

Explicitly, we fix the following leakage functions. Set
\begin{equation}
    T_0 = \lbrace 0, \pm 2 \rbrace, \quad T_1 = \lbrace 0, \pm 1 \rbrace, \quad T_2 = \lbrace 0, \pm 3 \rbrace, \quad T_3 = \lbrace 0, \pm 2 \rbrace, \quad T_4 = \lbrace 0, \pm 1 \rbrace
\end{equation}
and associate to each shareholder indexed by $i = 0, 1, 2, 3, 4$ the leakage function
\begin{equation}
    \lambda_i: \mathbb{F}_7 \to \mathbb{F}_2, \; x \mapsto \begin{cases}
        0 & x \in T_i\\
        1 & \text{else}.
    \end{cases}
\end{equation}
To each line $f(x) \in \mathbb{F}_7^\ast[x]$ with non-zero coefficients, we may then associate the leakage transcript
\begin{equation}
    \mathbf{t}_f \defeq \left( \lambda_0(f(0)), \lambda_1(f(1)), \ldots, \lambda_4(f(4)) \right) \in \mathbb{F}_2^5.
\end{equation}
It is easily verified (with a computer!) that given lines $f(x), \ell(x) \in \mathbb{F}_7^\ast[x]$, $\mathbf{t}_f = \mathbf{t}_\ell$ implies $g_{0,1}(\mathbf{f}) = g_{0,1}(\boldsymbol\ell)$, which implies that this collection of leakage functions yields a 5-bit secret reconstruction scheme for a $k=2$ instance of multiplicative Shamir over $\mathbb{F}_7^\ast$. 

\begin{figure}
    \centering
    \begin{equation}
\begin{array}{c|ccccccc}
 & B_0 & B_1 & B_2 & B_3 & B_4 & B_5 & B_6 \\
\hline
0 & \mathbb{F}_7 & \mathbb{F}_7^\ast & \mathbb{F}_7^\ast & \mathbb{F}_7^\ast & \mathbb{F}_7^\ast & \mathbb{F}_7^\ast & \mathbb{F}_7^\ast \\
1 & \mathbb{F}_7 & \pm 1, \pm 2 & \pm 1, \pm 3 & 0, \pm 3 & \pm 2, \pm 3 & 0, \pm 1 & 0, \pm 2 \\
2 & \mathbb{F}_7 & \pm 1, \pm 3 & \pm 2, \pm 3 & 0, \pm 2 & \pm 1, \pm 2 & 0, \pm 3 & 0, \pm 1 \\
3 & \mathbb{F}_7 & 0, \pm 3 & 0, \pm 2 & \pm 1, \pm 3 & 0, \pm 1 & \pm 1, \pm 2 & \pm 2, \pm 3 \\
4 & \mathbb{F}_7 & \pm 2, \pm 3 & \pm 1, \pm 2 & 0, \pm 1 & \pm 1, \pm 3 & 0, \pm 2 & 0, \pm 3 \\
5 & \mathbb{F}_7 & 0, \pm 1 & 0, \pm 3 & \pm 1, \pm 2 & 0, \pm 2 & \pm 2, \pm 3 & \pm 1, \pm 3 \\
6 & \mathbb{F}_7 & 0, \pm 2 & 0, \pm 1 & \pm 2, \pm 3 & 0, \pm 3 & \pm 1, \pm 3 & \pm 1, \pm 2 \\
\end{array}
\end{equation}
    \caption{Images of $B_\gamma \subseteq \mathbb{F}_7[x]$ at an evaluation point $\alpha \in \mathbb{F}_7$: given a row indexed by $\alpha \in \mathbb{F}_7$ and a column indexed by $B_\gamma$, the table entry at $(\alpha, B_\gamma)$ is $B_\gamma(\alpha)\subseteq \mathbb{F}_q$.}
    \label{fig: F7 leak table}
\end{figure}

Furthermore, it may be interesting to observe that even a \textit{single bit of leakage} is sufficient to learn some information about the secret. In Figure \ref{fig: F7 leak table}, we can see that it's easy to construct 1-bit leakage queries that eliminate at least one coefficient product from consideration. As a concrete example, suppose we query a server indexed by $1 \in \mathbb{F}_7$, who holds the codeword symbol $f(1)$ for some line $f(x) \in \mathbb{F}_7[x]$. We request of that server to evaluate the following leakage function on their codeword symbol:
\begin{equation}
    \lambda_1: \mathbb{F}_7 \to \mathbb{F}_2, \quad x \mapsto \begin{cases}
        0 & x \in \lbrace 0, \pm 1 \rbrace\\
        1 & \text{else}.
    \end{cases}
\end{equation}
If we receive a ``0'' in response, then we may immediately observe that $f(x) \not \in B_4$, since $B_4(1) = \lbrace \pm 2, \pm 3\rbrace$. On the other hand, if we receive a ``1'' in response, then we may immediately observe that $f(x) \not \in B_5$. Thus, in a \textit{strict} sense of information-theoretic security, we observe that ``multiplicative Shamir'' is not 1-bit leakage resilient over $\mathbb{F}_7$ when $k=2$. It may be interesting to study if or how this behavior scales outside the toy parameters of this example.

\end{document}